\documentclass{article}

\usepackage[preprint]{neurips_2026}

\usepackage[utf8]{inputenc}
\usepackage[T1]{fontenc}
\usepackage{hyperref}
\usepackage{url}
\usepackage{booktabs}
\usepackage{longtable}
\usepackage{wrapfig}
\usepackage{float}
\usepackage{subcaption}
\usepackage{amsmath,amssymb,amsthm,mathtools}
\usepackage{graphicx}
\usepackage{microtype}
\usepackage{enumitem}
\usepackage{nicefrac}
\usepackage{xcolor}
\usepackage{array}
\usepackage{bbm}
\usepackage{algorithm}
\usepackage{algpseudocode}

\hypersetup{colorlinks=true,linkcolor=blue,citecolor=blue,urlcolor=blue}

\newtheorem{assumption}{Assumption}

\newtheorem{proposition}{Proposition}
\newtheorem{theorem}{Theorem}
\newtheorem{corollary}{Corollary}
\newtheorem{lemma}{Lemma}
\theoremstyle{remark}
\newtheorem{remark}{Remark}

\newcommand{\DR}{\mathcal D^r}
\newcommand{\DO}{\mathcal D^o}
\newcommand{\D}{\mathcal D}
\newcommand{\R}{\mathbb R}
\newcommand{\E}{\mathbb E}

\newcommand{\Prob}{\mathbb P}
\newcommand{\one}{\mathbbm 1}

\newcommand{\calZ}{\mathcal Z}

\newcommand{\calG}{\mathcal G}
\newcommand{\ipm}{\operatorname{IPM}}

\newcommand{\expit}{\operatorname{expit}}
\newcommand{\Normal}{\mathcal N}
\newcommand{\Bern}{\operatorname{Bernoulli}}
\newcommand{\KL}{\operatorname{KL}}

\newcommand{\ATE}{\operatorname{ATE}}
\newcommand{\ess}{n_{\mathrm{eff}}}

\title{B-CALM: Bias-Limited Bayesian Borrowing \\for RCT-Anchored Treatment Effects \\under Covariate Mismatch}

\author{%
  Amir Asiaee\thanks{Corresponding author: \texttt{amir.asiaeetaheri@vumc.org}.} \\
  Department of Biostatistics\\
  Vanderbilt University Medical Center\\
  Nashville, TN 37203, USA \\
  \And
  Samhita Pal \\
  Department of Biostatistics\\
  Vanderbilt University Medical Center\\
  Nashville, TN 37203, USA \\
}

\begin{document}
\maketitle

\begin{abstract}
Randomized controlled trials (RCTs) identify treatment effects in the randomized trial population but are often too small for reliable heterogeneity estimation; observational studies (OS) are larger but confounded and measured on only partially overlapping covariates. We develop Bayesian Calibrated ALignment under covariate Mismatch (B-CALM), a Bayesian borrowing method for RCT-defined conditional average treatment effect (CATE) estimation. B-CALM maps source-specific covariates into a shared latent state, jointly models trial and observational outcome surfaces, and uses baseline-bias and comparative-bias functions to represent how the OS departs from the trial estimand. The comparative-bias prior becomes an explicit sensitivity knob: we prove a finite-feature bias-limited information bound showing that observational contrast information about the trial treatment-effect function is capped by the prior precision of this bias function, and derive an effective-sample-size formula showing that the RCT-equivalent information contributed by the OS saturates as OS sample size grows. The theory also combines a PAC-Bayes trial-risk bound with an integral-probability-metric (IPM) alignment and calibration decomposition that separates RCT empirical risk, latent alignment, and residual calibration of the debiased OS surface. In synthetic, semi-synthetic, and pediatric-obesity external-control studies, B-CALM maintains near-nominal average coverage and low negative transfer while pooled and causal-forest baselines can become overconfident under comparative bias.
\end{abstract}

\section{Introduction}
\label{sec:intro}

Randomized controlled trials (RCTs) remain the main design for internally valid treatment-effect estimation. In many clinical and policy settings, however, RCTs are powered for an average treatment effect, not for conditional average treatment effects (CATEs) over patient subgroups. Large observational studies (OS) offer the obvious source of additional information: they contain many more individuals, often measured through detailed health systems, registries, or administrative databases. That information carries its own bias. Treatment assignment is not randomized, unmeasured confounding can distort treatment contrasts, patient populations can differ from the trial, and the measured covariates in the two sources often overlap only partially.

The problem is therefore not simply how to combine an RCT and an OS. The harder task is to let the OS reduce uncertainty about the \emph{trial-population} CATE while preventing a large but biased source from becoming the de facto estimand. Recent calibrated-borrowing methods were built around this constraint. R-OSCAR uses observational outcome predictions as nuisance functions and calibrates their residual discrepancy in the trial \citep{asiaee2025roscar}; MR-OSCAR extends the same calibration strategy by imputing trial-missing observational covariates \citep{pal2026mroscar}; and CALM avoids explicit imputation by aligning source-specific embeddings before calibration \citep{asiaee2026calm}. These methods use the RCT to define and identify the target contrast, but primarily return point estimates. Bayesian borrowing takes a different route: power priors, commensurate priors, robust meta-analytic-predictive priors, and joint external-control models use priors or explicit bias parameters to regulate external evidence \citep{ibrahim2000power,hobbs2011hierarchical,schmidli2014robust,harrell2023hxcontrol}, but usually target scalar or low-dimensional treatment summaries, such as average effects, response rates, or prespecified subgroup effects, rather than a covariate-indexed CATE function under covariate mismatch.

\vspace{-0.6em}
\paragraph{Modeling principle.}
Bayesian Calibrated ALignment under covariate Mismatch (B-CALM) views the observational study as a biased, partially mismatched measurement of the trial outcome surfaces. Let trial covariates be $X^r=(W,U)$ and observational covariates be $X^o=(W,V)$, where $W$ is shared, $U$ is trial-only, and $V$ is observational-only. Source-specific encoders map each covariate block into a common latent patient state $Z$. The latent state is not itself the estimand; it is the coordinate system for comparing the two data sources at matched patient states.

On this shared coordinate system, the RCT defines the trial outcome surface
\begin{equation*}
  \eta^r(a,z)=\mu(z)+a\tau(z).
\end{equation*}
Here $\mu(z)$ is the trial control-outcome surface, the prognosis at $A=0$, and $\tau(z)$ is the trial treatment-contrast surface. Thus $\eta^r(0,z)=\mu(z)$ and $\eta^r(1,z)-\eta^r(0,z)=\tau(z)$; learning $\tau$ is the CATE problem after integrating over the trial encoder. Both $\mu$ and $\tau$ are unknown functions because the RCT is finite and covers patient states sparsely.

The observational study is allowed to have a different surface:
\begin{equation*}
  \eta^o(a,z)=\mu(z)+a\tau(z)+b_0(z)+a b_\Delta(z).
\end{equation*}
The function $b_0(z)$ is the \emph{baseline bias}: at the same latent patient state, it is the difference between OS and RCT controls. The function $b_\Delta(z)$ is the \emph{comparative bias}: it is the difference between the OS treatment contrast and the RCT treatment contrast after the baseline shift is removed. Equivalently, the OS contrast is $\tau(z)+b_\Delta(z)$, while the trial contrast of interest remains $\tau(z)$. This decomposition is the causal accounting of B-CALM: the OS may inform prognosis, representation, and even treatment contrasts, but only through an explicit model for how it may be biased relative to the RCT.

\vspace{-0.6em}
\paragraph{Borrowing as sensitivity.}
Priors on $b_0$ and $b_\Delta$ control borrowing. A tight prior on $b_\Delta$ encodes the belief that OS and RCT treatment contrasts are commensurate at each latent state, so the OS can sharpen inference for $\tau$. A diffuse prior allows the OS treatment contrast to be arbitrarily biased, so comparative borrowing reverts toward RCT-only inference. This is a function-valued analogue of commensurate and power-prior borrowing \citep{ibrahim2000power,hobbs2011hierarchical,hobbs2012commensurate}: instead of discounting the entire external likelihood by one scalar, B-CALM places a prior on the covariate-indexed discrepancy that determines which parts of the OS contrast are allowed to inform the trial contrast.

\vspace{-0.6em}
\paragraph{Main claim.}
The key consequence is \emph{bias-limited borrowing}. If the comparative-bias prior has nonzero variance, observational information about the trial treatment-effect function has a finite ceiling, no matter how large the OS becomes. In a linearized function model, the OS Fisher information for $\tau$ is bounded by the prior precision of $b_\Delta$; in the scalar special case for a one-number treatment summary, the effective sample size (ESS) saturates at a level determined by the ratio of observational noise to prior comparative-bias variance. This makes sensitivity analysis concrete: changing the prior standard deviation for comparative bias changes the maximum amount of treatment-contrast information the OS is allowed to contribute.

\paragraph{Contributions and positioning.}
We make four contributions.
\begin{enumerate}[leftmargin=*,itemsep=2pt,topsep=2pt]
\item We develop \textbf{Bayesian calibrated alignment} for RCT-defined CATE estimation in the partially overlapping-covariate setting, yielding a posterior over encoders, outcome surfaces, bias functions, and trial-population CATEs.
\item We introduce an explicit \textbf{baseline-bias} and \textbf{comparative-bias} decomposition that separates the trial estimand $\tau$ from the observational contrast $\tau+b_\Delta$ after alignment onto the shared latent representation.
\item We prove \textbf{function-valued bias-limited borrowing}: the comparative-bias prior precision bounds the information observational contrasts can contribute about the trial treatment-effect surface, with a scalar ESS corollary for one-number treatment summaries and the limiting full-pooling and RCT-only regimes.
\item We give a variational implementation and evaluate it on synthetic, semi-synthetic, and pediatric-obesity external-control studies. Across the main factorial simulation study, B-CALM has near-nominal 90\% credible-interval coverage (0.939 on average, with a minimum of about 0.86 under the most adversarial comparative-bias settings) with low negative transfer (0.023), while pooled and causal-forest baselines become overconfident under comparative bias; in the real-data study, B-CALM achieves a modest width reduction while keeping sensitivity to external-control bias explicit.
\end{enumerate}

Relative to R-OSCAR, MR-OSCAR, and CALM \citep{asiaee2025roscar,pal2026mroscar,asiaee2026calm}, B-CALM casts calibrated borrowing as posterior sensitivity analysis, rather than a point estimate with bootstrap or other separately propagated uncertainty. Relative to power, commensurate, robust MAP, and dynamic borrowing priors \citep{ibrahim2000power,hobbs2011hierarchical,schmidli2014robust,kaizer2018bayesian,kotalik2021framework}, B-CALM does not introduce a single global borrowing weight, such as a likelihood discount or exchangeability indicator; the amount of borrowing is implied locally by the prior and posterior uncertainty in the bias function $b_\Delta(z)$, learned jointly with source-specific encoders. Relative to representation-based CATE and heterogeneous-feature transfer methods \citep{johansson2016learning,shalit2017estimating,bica2022heterogeneous}, the relevant shift is RCT-versus-OS rather than treated-versus-control, and the observational source is a biased proxy for the trial estimand rather than its own target population. Generalizing from the randomized-trial population to a broader target population is a distinct transportability problem, reviewed by \citet{colnet2024causal} and \citet{degtiar2023review}. More recently, it has been studied under outcome-distribution shift using several sensitivity-analysis frameworks, including the marginal sensitivity model \citep{asiaee2026sharp} and the omitted-variable-bias framework \citep{asiaee2026ovb}.

\vspace{-0.6em}
\section{Related work}
\vspace{-0.6em}
\paragraph{Observational borrowing for trial estimands.}
The closest line of work starts from observational outcome models to improve trial-based CATE estimation while using the randomized trial to define and identify the RCT-defined CATE. R-OSCAR calibrates observational conditional-mean predictions using RCT data, reducing variance while preserving unbiased trial contrasts under regularity conditions \citep{asiaee2025roscar}. MR-OSCAR carries the calibration idea to covariate mismatch by imputing observational-only variables in the trial and deriving a finite-sample decomposition with an imputation-error term \citep{pal2026mroscar}. CALM, the direct predecessor of B-CALM, replaces explicit imputation with source-specific embeddings aligned into a shared representation and decomposes risk into alignment, outcome-model, and calibration terms \citep{asiaee2026calm}. B-CALM inherits this calibrated-alignment target and the same estimand discipline: observational information can reduce uncertainty only through calibration or explicit bias modeling, not by redefining the RCT estimand toward the OS contrast. The change is inferential rather than estimand-level: B-CALM turns calibrated alignment from a point-estimation procedure with separately propagated uncertainty into a Bayesian joint model, yielding a posterior over encoders, trial and observational outcome surfaces, source-bias functions, and trial-population CATEs.

\vspace{-0.6em}
\paragraph{Bayesian borrowing and external controls.}
B-CALM's borrowing mechanism sits nearer to Bayesian external-data methods, which regulate how much historical or external evidence enters an analysis. Power priors downweight the historical likelihood by a scalar \citep{ibrahim2000power}; commensurate priors borrow according to parameter similarity \citep{hobbs2011hierarchical}; robust meta-analytic-predictive priors add a weakly informative mixture component to protect against prior-data conflict \citep{neuenschwander2010summarizing,schmidli2014robust}; and multisource exchangeability models use discrete indicators to decide which sources are exchangeable \citep{kaizer2018bayesian}, with heterogeneous-effect dynamic borrowing studied in \citet{kotalik2021framework}. These methods answer the question of \emph{how much} external evidence to use, typically for average effects or matched-covariate models; reviews by \citet{vanrosmalen2018including} and \citet{lewis2019borrowing} emphasize the danger of uncontrolled prior-data conflict. B-CALM asks the borrowing question pointwise in latent state: how much should the OS treatment contrast at state $z$ inform the posterior for the RCT treatment contrast at the same state? The prior on $b_\Delta(z)$ answers that question: tight priors permit local pooling, while diffuse priors return the contrast toward RCT-only inference. This is a function-valued version of the explicit bias parameters advocated in joint external-control modeling \citep{harrell2023hxcontrol}.

\vspace{-0.6em}
\paragraph{Representation learning under treatment and domain mismatch.}
Representation-based CATE methods such as TARNet and CFRNet \citep{johansson2016learning,shalit2017estimating} address treatment-arm imbalance within an observational study: they learn a representation $\Phi(X)$ and outcome heads for $Y(0)$ and $Y(1)$, with balance penalties or bounds involving an integral-probability metric between $P\{\Phi(X)\mid A=1\}$ and $P\{\Phi(X)\mid A=0\}$. HTCE \citep{bica2022heterogeneous} addresses a domain-transfer version of CATE estimation in which source and target domains may have different covariate spaces, using shared and private representations to transfer potential-outcome information. B-CALM uses representation learning for a different purpose: the mismatch is between randomized and observational sources, the estimand remains the RCT-defined CATE, and OS treatment-contrast information is allowed to help only through the explicit comparative-bias function $b_\Delta(z)$.

\paragraph{CATE uncertainty and RCT--OS fusion.}
Bayesian CATE learners provide posterior uncertainty for treatment-effect surfaces, including multi-task Gaussian processes for individualized effects \citep{alaa2017bayesian} and Bayesian causal forests \citep{hahn2020bayesian}. These methods model CATE uncertainty, but they do not by themselves solve the problem of borrowing from a potentially biased OS under covariate mismatch. A separate RCT--OS fusion literature uses trial data to correct or constrain observational CATE learning: \citet{kallus2018removing} use limited experimental data to correct hidden confounding in observational models, \citet{hatt2022combining} use observational data to learn shared representations and randomized data to learn trial-specific structure, and \citet{ilse2023combining} combine observational and interventional distributions through causal reductions. Closest to our Bayesian framing, Causal-ICM \citep{dimitriou2024causalicm} uses multi-task Gaussian processes with a borrowing parameter for RCT--OS heterogeneous-treatment-effect estimation. B-CALM differs by targeting the RCT-defined CATE under partially overlapping covariates and by making source bias an explicit baseline and comparative-bias process whose prior controls local borrowing.

\vspace{-0.6em}
\section{Problem setup}
\vspace{-0.6em}

We observe two samples: a randomized trial $\DR=\{(Y_i^r,A_i^r,W_i^r,U_i^r):i=1,\ldots,n^r\}$ and an observational study $\DO=\{(Y_j^o,A_j^o,W_j^o,V_j^o):j=1,\ldots,n^o\}$. The treatment $A\in\{0,1\}$ is randomized in $\DR$ but not in $\DO$. The covariate block $W$ is shared, $U$ is trial-only, and $V$ is observational-only. We use $s\in\{r,o\}$ for source. The source indicator is written as a superscript throughout ($n^r$, $X^r$, $\eta^o$, $R^r$); the only exception is the residual variances $\sigma_r^2,\sigma_o^2$.

Let $Y^r(a)$ denote the trial potential outcome under treatment $a$. We focus on continuous outcomes. The trial estimand is
\begin{equation}
  \tau^r(x^r)=\E\{Y^r(1)-Y^r(0)\mid X^r=x^r\}.
  \label{eq:continuous_estimand}
\end{equation}
Binary outcomes can be handled by placing the same bias decomposition on the linear-predictor scale of a generalized linear model and transforming posterior draws back to risks and risk differences; Appendix~\ref{app:binary-outcomes} provides the notation. The OS superscript denotes a data source, not a second causal estimand: because $A^o$ is not randomized, OS treatment-level outcome surfaces are treated as observed conditional surfaces that may be biased relative to the trial causal surface.

\begin{assumption}[Trial identification]
\label{ass:trial-identification}
Within the RCT, treatment is randomized and satisfies positivity: $0<\Prob(A^r=1\mid X^r)<1$. Consistency holds: $Y^r=Y^r(A^r)$.
\end{assumption}

\begin{assumption}[Trial-relevant latent state]
\label{ass:trial-latent-state}
Write $m_a^r(x^r)=\E\{Y^r(a)\mid X^r=x^r\}$. There exist a trial encoder $q_0^r(z\mid x^r)$, functions $\mu_0,\tau_0:\calZ\to\R$, and a representation error $\varepsilon_{\mathrm{rep}}\ge0$ such that, for $a\in\{0,1\}$,
\begin{equation}
  \sup_{x^r}
  \left|
    m_a^r(x^r)
    -
    \E_{Z\sim q_0^r(\cdot\mid x^r)}
    \{\mu_0(Z)+a\tau_0(Z)\}
  \right|
  \le
  \varepsilon_{\mathrm{rep}}.
  \label{eq:trial-latent-state-assumption}
\end{equation}
\end{assumption}
This assumption says that the latent state is a sufficient working coordinate system for the RCT outcome surfaces, up to approximation error. It is not an assumption that the OS is causal, and it is not an assumption that $Z$ is a literal biological hidden variable. In the deterministic-encoder case, Eq.~\eqref{eq:trial-latent-state-assumption} reduces to $m_a^r(x^r)\approx\mu_0\{e^r(x^r)\}+a\tau_0\{e^r(x^r)\}$; in the variational-encoder case, the approximation is averaged over the encoder distribution.

Alignment can be encouraged in several ways. If the shared covariate block $W$ is informative, learned encoders can be constrained to preserve it, for example by adding an auxiliary decoder that predicts $W$ from $Z$. If linked records are available, such as trial participants linked to external registry or EHR records, one can directly penalize disagreement between the two encoded views of the same unit. In the theory, we use an integral probability metric (IPM) to quantify the remaining difference between the induced trial and OS latent distributions. For probability distributions $P_1$ and $P_2$ on $\calZ$ and a discriminator class $\calG$, define:
\begin{equation}
  \ipm_{\calG}(P_1,P_2)
  =
  \sup_{g\in\calG}
  \left|
    \E_{Z\sim P_1}\{g(Z)\}
    -
    \E_{Z\sim P_2}\{g(Z)\}
  \right|.
  \label{eq:ipm-definition}
\end{equation}
In a learned-encoder implementation, the IPM could also be used as a regularizer. Maximum mean discrepancy is one special case, obtained when $\calG$ is the unit ball of a reproducing-kernel Hilbert space \citep{gretton2012kernel}. The assumption below records the residual source mismatch after the chosen encoder construction or regularization; it applies to any encoder construction that satisfies the stated IPM bound.

\begin{assumption}[Calibratable source alignment]
\label{ass:source-alignment}
Let $P_Z^r(\phi^r)$ and $P_Z^o(\phi^o)$ denote the latent distributions induced by $Z^r\sim q_{\phi^r}(\cdot\mid X^r)$ and $Z^o\sim q_{\phi^o}(\cdot\mid X^o)$. For a discriminator class $\calG$, define
\begin{equation}
  \Delta_Z(\phi^r,\phi^o)
  =
  \ipm_{\calG}\{P_Z^r(\phi^r),P_Z^o(\phi^o)\}.
  \label{eq:latent-alignment-assumption}
\end{equation}
There exist source-specific encoders $(\phi_0^r,\phi_0^o)$, with $\phi_0^r$ satisfying Assumption~\ref{ass:trial-latent-state}, such that
\(
  \Delta_Z(\phi_0^r,\phi_0^o)
  \le
  \varepsilon_{\mathrm{align}}
\)
for some residual alignment error $\varepsilon_{\mathrm{align}}$.
\end{assumption}
Here ``source-specific'' means that the two latent distributions are generated from different observed covariate systems and different encoders: trial units pass through $q_{\phi^r}$ and OS units pass through $q_{\phi^o}$. Assumption~\ref{ass:source-alignment} is therefore the mathematical version of ``the two covariate systems can be put on a common latent coordinate system, but not perfectly.'' The remaining mismatch is not ignored; it appears as the IPM term $\Delta_Z$. Small alignment error makes OS information more useful, while large alignment error pushes the analysis toward RCT-only information.

We make no ignorability assumption for the observational treatment assignment. Instead, the observational conditional outcome surface is allowed to differ from the trial causal surface through bias functions. This separates B-CALM from point-identification designs that presume observational treatment comparisons are causal: the observational data are useful but not presumed causal for treatment comparisons.

\section{B-CALM model}
\label{sec:model}
\vspace{-0.6em}
\begin{figure}[H]
  \centering
  \includegraphics[width=0.9\textwidth]{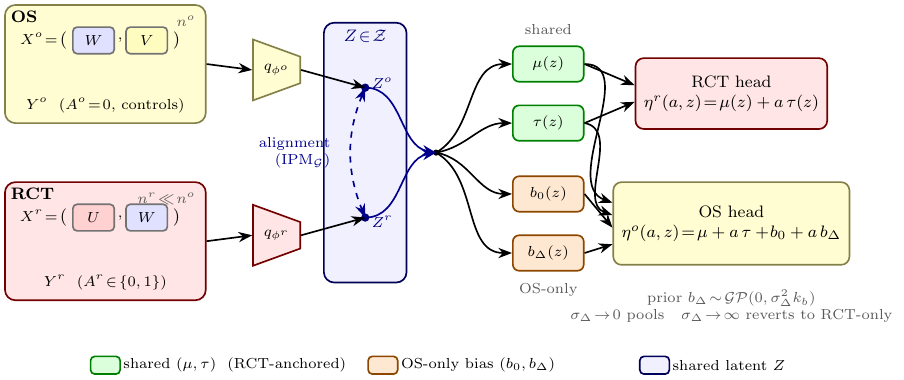}
  \caption{B-CALM architecture and causal bookkeeping. Source-specific encoders $q_{\phi^r}$ and $q_{\phi^o}$ map the partially overlapping covariate blocks $X^r=(W,U)$ and $X^o=(W,V)$ into a shared latent state $Z$, with source alignment encoded by the chosen latent maps and, when those maps are learned, regularized by an IPM or related calibration loss. The RCT head defines the RCT-defined functions $(\mu,\tau)$, while the OS head observes the shifted functions $(\mu+b_0,\tau+b_\Delta)$. Priors on the bias functions, especially $b_\Delta$, determine how much OS contrast information can enter the posterior for the trial CATE.}
  \label{fig:bcalm-arch}
\end{figure}

\subsection{Probabilistic calibrated alignment}

Figure~\ref{fig:bcalm-arch} summarizes our framework. B-CALM introduces source-specific encoders $Z_i^r\sim q_{\phi^r}(z\mid W_i^r,U_i^r)$ and $Z_j^o\sim q_{\phi^o}(z\mid W_j^o,V_j^o)$, where $q_{\phi^r}$ and $q_{\phi^o}$ may be deterministic neural encoders, variational encoders, sparse linear projections, or GP latent-variable maps. The corresponding source-specific latent distributions are the distributions of $Z_i^r$ when trial covariates are passed through $q_{\phi^r}$ and of $Z_j^o$ when OS covariates are passed through $q_{\phi^o}$. The encoder does not make the OS causal; it defines a common coordinate system on which trial and OS outcome surfaces can be compared, while the remaining source discrepancy is absorbed by explicit bias functions. In the general framework, learned encoders can be regularized by matching latent distributions, for example through the IPM in Eq.~\eqref{eq:ipm-definition}; shared-covariate reconstruction can be added if one specifies an auxiliary decoder.

We define the outcome model on the shared latent space. For the continuous-outcome model used in the paper, the likelihood is Gaussian; write $p_{\mathrm G}(y;m,\sigma^2)$ for the $\Normal(m,\sigma^2)$ density. The RCT likelihood is
\begin{equation}
  Y_i^r\mid A_i^r,Z_i^r,\mu,\tau \sim \Normal\!\left\{\eta^r(A_i^r,Z_i^r),\sigma_r^2\right\},
  \quad
  \eta^r(a,z)=\mu(z)+a\tau(z).
  \label{eq:rct_eta}
\end{equation}
The OS likelihood has the same Gaussian form but with source-bias functions:
\begin{equation}
  Y_j^o\mid A_j^o,Z_j^o,\mu,\tau,b_0,b_\Delta \sim \Normal\!\left\{\eta^o(A_j^o,Z_j^o),\sigma_o^2\right\},
  \quad
  \eta^o(a,z)=\mu(z)+a\tau(z)+b_0(z)+a b_\Delta(z).
  \label{eq:os_eta}
\end{equation}
The function $\mu$ is the trial baseline outcome surface and $\tau$ is the trial treatment-effect function. The two source-specific bias functions record how the OS departs from the trial. Evaluating Eq.~\eqref{eq:os_eta} at $a=0$ and $a=1$, the control-surface gap $\eta^o(0,z)-\eta^r(0,z)=b_0(z)$ is the \textbf{baseline bias} (how OS controls differ from RCT controls at the same latent $z$, through unobserved comorbidities, selection, or measurement drift) and the contrast gap $[\eta^o(1,z)-\eta^o(0,z)]-[\eta^r(1,z)-\eta^r(0,z)]=b_\Delta(z)$ is the \textbf{comparative bias} (the residual that confounding leaves on the contrast after the baseline shift is removed). The OS has its own pair $(\mu+b_0,\tau+b_\Delta)$ for the control surface and treatment effect, while the RCT identifies $(\mu,\tau)$. The posterior CATE is computed by integrating the encoder:
\begin{equation}
  \tau^r(x^r)=\E\left[\tau(Z)\mid Z\sim q_{\phi^r}(\cdot\mid x^r),\D\right],
  \label{eq:posterior_cont}
\end{equation}

\paragraph{Priors and borrowing interpretation.}
One nonparametric specification places independent GP priors $\mu\sim\mathcal{GP}(0,k_\mu)$, $\tau\sim\mathcal{GP}(0,k_\tau)$, $b_0\sim\mathcal{GP}(0,\sigma_0^2 k_b)$, and $b_\Delta\sim\mathcal{GP}(m_\Delta,\sigma_\Delta^2 k_b)$. The mean $m_\Delta$ can encode directional prior knowledge about confounding. The variance $\sigma_\Delta^2$ controls comparative borrowing. If $\sigma_\Delta^2=0$ and $m_\Delta=0$, the observational and randomized treatment contrasts are pooled directly. If $\sigma_\Delta^2\to\infty$, the OS treatment contrast is allowed to differ arbitrarily from the RCT contrast and therefore contributes essentially no comparative information about $\tau$. Intermediate values define a sensitivity analysis.

This prior standard deviation connects B-CALM to classical Bayesian borrowing. A power prior discounts an external likelihood globally; a commensurate prior borrows according to similarity between low-dimensional parameters. B-CALM instead places the commensurability question on the contrast discrepancy $b_\Delta(z)$ itself. The OS can be highly precise about its own contrast $\tau(z)+b_\Delta(z)$, but its information about the trial contrast $\tau(z)$ is limited by how tightly the analyst is willing to concentrate $b_\Delta$ around its prior mean.

\paragraph{Variational implementation.}
For computation, we approximate the posterior over functions by basis-function or neural-network parameterizations, with weight vectors $\omega_\mu,\omega_\tau,\omega_0,\omega_\Delta$ for $\mu,\tau,b_0,b_\Delta$ respectively. Let $\theta=(\omega_\mu,\omega_\tau,\omega_0,\omega_\Delta,\phi^r,\phi^o)$ collect these weights together with the encoder parameters. Given priors $p(\theta)$, a variational approximation $q_\lambda(\theta)$ maximizes
\begin{align}
  \mathcal L(\lambda)
  ={}&\E_{q_\lambda}\left[
      \sum_{i=1}^{n^r}\log p_{\mathrm G}\{Y_i^r;\eta^r(A_i^r,Z_i^r),\sigma_r^2\}
      +\sum_{j=1}^{n^o}\log p_{\mathrm G}\{Y_j^o;\eta^o(A_j^o,Z_j^o),\sigma_o^2\}
  \right] \nonumber\\
  &-\KL\{q_\lambda(\theta)\|p(\theta)\}
  -\lambda_{\mathrm{align}}\,\widehat{\mathcal A}(q_{\phi^r},q_{\phi^o}),
  \label{eq:elbo}
\end{align}
where $Z_i^r\sim q_{\phi^r}(\cdot\mid W_i^r,U_i^r)$ and $Z_j^o\sim q_{\phi^o}(\cdot\mid W_j^o,V_j^o)$. For learned-encoder implementations, the alignment penalty $\widehat{\mathcal A}$ may be an MMD penalty, reconstruction loss for shared covariates, adversarial source-classification loss, or a combination; in the fixed-basis experiments, this term is absent because the representation is chosen before fitting. We compute the posterior by Monte Carlo draws of $(\theta,Z)$ and transformation through Eq.~\eqref{eq:posterior_cont}.

\paragraph{Algorithm sketch.}
B-CALM follows five steps: (1) choose an encoder family and, when encoders are learned, an alignment loss; (2) choose priors for $\mu,\tau,b_0,b_\Delta$, including prespecified sensitivity values for $\sigma_\Delta$; (3) optimize the evidence lower bound (ELBO) in Eq.~\eqref{eq:elbo}; (4) draw posterior samples of the trial control surface, treated surface, and CATEs for trial-population units; and (5) report posterior summaries and sensitivity curves over $\sigma_\Delta$.

\section{Theory}
\label{sec:theory}

B-CALM does not claim that observational data identify the RCT CATE. Assumption~\ref{ass:trial-identification} fixes the RCT estimand throughout; Assumption~\ref{ass:trial-latent-state} justifies writing the trial outcome surfaces on $Z$; and Assumption~\ref{ass:source-alignment} enters the IPM-calibration risk decomposition through $\Delta_Z$. The information bounds below condition on fixed encoders and are algebraic consequences of the Gaussian contrast model and the comparative-bias prior. The section proceeds in four steps. We first study the two extremes of the comparative-bias prior: full pooling when the OS and RCT treatment contrasts are assumed identical, and RCT-only inference when the OS contrast may differ arbitrarily. We then collapse the estimand to a one-number treatment-effect summary, such as an ATE, to derive the effective-sample-size calculation behind bias-limited borrowing. Next, we extend the same calculation to finite-dimensional representations of the CATE function, showing that the comparative-bias prior precision caps OS information about the RCT treatment-effect surface. Finally, we state a calibrated-alignment risk decomposition that separates RCT empirical risk, latent alignment, and residual calibration of the debiased OS surface.

\paragraph{Borrowing limit cases.}
We first study the two extremes of the comparative-bias prior.

\begin{proposition}[Full-pooling and RCT-only limits]
\label{prop:limits}
Consider the Gaussian B-CALM model in Eqs.~\eqref{eq:rct_eta}--\eqref{eq:os_eta} with fixed encoders and proper priors on $\mu$ and $\tau$. If $b_\Delta\equiv 0$, the OS and RCT share the same treatment-effect surface and the OS likelihood contributes directly to the posterior precision for $\tau$. If the prior on $b_\Delta$ converges to a flat Gaussian process prior under a uniform flat-precision condition on the OS contrast directions, the OS likelihood identifies only $\tau+b_\Delta$ and contributes no finite information to $\tau$ beyond information about shared baseline and representation parameters.
\end{proposition}

\paragraph{Effective sample size for an ATE.}
For intuition and comparison with classical borrowing, first collapse the estimand to a one-number treatment-effect summary, such as the trial-population ATE $\bar\tau=\E_{X^r}\{\tau^r(X^r)\}$. This is not the main B-CALM estimand; it is the finite-dimensional special case where effective-sample-size language is most transparent. The same calculation applies to any fixed finite-dimensional summary of the CATE surface, such as a prespecified subgroup treatment effect.

\begin{corollary}[ATE special case: bias-limited effective sample size]
\label{thm:scalar}
Take the one-dimensional special case $\widehat{\bar\tau}^r\mid\bar\tau\sim\Normal(\bar\tau,\sigma_r^2/n^r)$ and $\widehat{\bar\tau}^o\mid\bar\tau,b\sim\Normal(\bar\tau+b,\sigma_o^2/n^o)$, with priors $b\sim\Normal(0,\sigma_b^2)$ and $\bar\tau\sim\Normal(m_0,v_0)$, where $\bar\tau$ is a scalar, one-number treatment contrast and $m_0,v_0$ are its prior mean and prior variance. Then the posterior precision for $\bar\tau$ is $V_n^{-1}=v_0^{-1}+n^r/\sigma_r^2+(\sigma_o^2/n^o+\sigma_b^2)^{-1}$, the OS contributes the same precision as an unbiased RCT contrast with effective sample size $\ess^o=n^o/(1+n^o\sigma_b^2/\sigma_o^2)$, and $\lim_{n^o\to\infty}\ess^o=\sigma_o^2/\sigma_b^2$.
\end{corollary}

Increasing $n^o$ can remove sampling noise in the OS contrast, but it cannot remove prior uncertainty about the OS-to-RCT contrast gap. The saturation level is determined by the ratio of observational noise to comparative-bias prior variance.

\paragraph{Function-valued bias-limited information.}
The one-number calculation becomes a B-CALM statement when applied to a latent treatment-effect surface. Here the encoder and function basis play different roles. The encoder parameters $\phi$ determine where a unit is placed in latent space; they answer the question, ``what is this unit's latent state $z$?'' The basis map $\psi$ is used after $z$ has been chosen; it answers the question, ``how do we represent functions such as $\tau(z)$ and $b_\Delta(z)$ on that latent space?'' The matrix $\Phi^o$ then stacks the OS basis rows $\psi(Z_j^o)^\top$ and is therefore a design matrix for the OS contrast regression, not an encoder.

\begin{theorem}[Function-valued bias-limited information]
\label{thm:functional}
Fix the encoders and hence the OS latent locations $Z_1^o,\ldots,Z_{n^o}^o$. Let $\psi(z)\in\R^d$ be a basis-function map, write $\tau(z)=\psi(z)^\top\beta_\tau$ and $b_\Delta(z)=\psi(z)^\top\beta_b$, and assume $\beta_b\sim\Normal(0,\Sigma_b)$ with $\Sigma_b$ positive definite. Stack the OS basis rows $\psi(Z_j^o)^\top$ into $\Phi^o\in\R^{n^o\times d}$. Suppose the OS contrast data are summarized by the Gaussian linear model $\widehat d^o=\Phi^o\beta_\tau+\Phi^o\beta_b+\varepsilon$, with $\varepsilon\sim\Normal(0,\sigma^2I_{n^o})$, where $\widehat d^o$ is a noisy vector of OS treatment-contrast measurements at the OS latent locations. Marginalizing over $\beta_b$ gives $\widehat d^o\mid\beta_\tau\sim\Normal\{\Phi^o\beta_\tau,\sigma^2I_{n^o}+\Phi^o\Sigma_b(\Phi^o)^\top\}$. If $\Phi^o$ has full column rank, the OS Fisher information for $\beta_\tau$ is
\begin{equation}
  \begin{aligned}
  I^o
  =
  (\Phi^o)^\top\!\left\{\sigma^2I+\Phi^o\Sigma_b(\Phi^o)^\top\right\}^{-1}\Phi^o
  =
  \left\{\Sigma_b+\sigma^2((\Phi^o)^\top\Phi^o)^{-1}\right\}^{-1}, \quad 
  I^o \preceq \Sigma_b^{-1} .
  \end{aligned}
\end{equation}
The positive-semidefinite statement means that the OS information is no larger in any finite-dimensional direction than the comparative-bias prior precision. As OS information grows in every basis direction, for example when the smallest eigenvalue of $(\Phi^o)^\top\Phi^o$ diverges, $I^o\to\Sigma_b^{-1}$.
\end{theorem}

The algebra above is standard Gaussian marginalization. The B-CALM implication is the important part: once the OS contrast is modeled as $\tau(z)+b_\Delta(z)$, the prior precision of the unknown function $b_\Delta$ is the information ceiling for borrowing about the unknown function $\tau$.

\paragraph{Encoder uncertainty in the ESS calculation.}
The preceding results condition on a fixed representation. When the OS encoder is uncertain, posterior variation in the mapped OS contrast acts like an additional irreducible discrepancy term. This is not the usual bias-variance tradeoff; it is extra uncertainty about how the OS contrast maps onto the RCT latent coordinate system.

\begin{proposition}[Encoder-aware effective sample size]
\label{prop:encoder-aware-ess}
Consider the scalar model $\widehat{\bar\tau}^r\mid\bar\tau\sim\Normal(\bar\tau,\sigma_r^2/n^r)$ and $\widehat{\bar\tau}^o\mid\bar\tau,b,\phi^o\sim\Normal\{\bar\tau+b+\xi(\phi^o),\sigma_o^2/n^o\}$, where $b\sim\Normal(0,\sigma_b^2)$ and $\xi(\phi^o)\sim\Normal(0,\sigma_{\mathrm{enc}}^2)$ are independent. Then the OS effective sample size becomes
\begin{equation}
  {\ess}^{o}_{\mathrm{enc}}
  =
  \frac{n^o}{1+n^o(\sigma_b^2+\sigma_{\mathrm{enc}}^2)/\sigma_o^2},
  \qquad
  \lim_{n^o\to\infty}{\ess}^{o}_{\mathrm{enc}}
  =
  \frac{\sigma_o^2}{\sigma_b^2+\sigma_{\mathrm{enc}}^2}.
\end{equation}
Encoder uncertainty is therefore additive to comparative-bias uncertainty in the borrowing cap.
\end{proposition}

\paragraph{Calibrated-alignment risk decomposition.}
The information bounds above explain the borrowing cap; they do not by themselves guarantee that a learned representation transfers trial risk well. The final result combines a PAC-Bayes bound for RCT empirical risk with a deterministic IPM-calibration decomposition for the debiased OS surface. Write $\rho$ for a posterior over full model draws $f=(\phi^r,\phi^o,\mu,\tau,b_0,b_\Delta)$ and $\rho_0$ for a data-independent prior over $f$. Let $\eta_f^r(a,z)=\mu_f(z)+a\tau_f(z)$ be the RCT surface under draw $f$, let $\eta_\star^r$ and $\eta_\star^o$ be the true RCT and OS conditional outcome surfaces on the latent space, and define the debiased OS mean under draw $f$ by $\eta_{\mathrm{deb},f}^o(a,z)=\eta_\star^o(a,z)-b_{0,f}(z)-a b_{\Delta,f}(z)$. Let $P_m$ denote the Gaussian outcome distribution indexed by mean $m$, define $\mathcal L_\ell(t,m)=\E_{Y\sim P_m}\ell(t,Y)$, and let $\pi^r(a\mid z)$ be the trial treatment distribution used to define the risk. For a fixed latent state $z$, set $r_f^r(z)=\sum_{a=0}^{1}\pi^r(a\mid z)\mathcal L_\ell\{\eta_f^r(a,z),\eta_\star^r(a,z)\}$ and $\widetilde r_f^o(z)=\sum_{a=0}^{1}\pi^r(a\mid z)\mathcal L_\ell\{\eta_f^r(a,z),\eta_{\mathrm{deb},f}^o(a,z)\}$. Also let $\widehat L^r(f)=(n^r)^{-1}\sum_{i=1}^{n^r}\ell\{\eta_f^r(A_i^r,Z_i^r),Y_i^r\}$. Then
\begin{equation}
  R^r(\rho)=\E_{f\sim\rho}\E_{Z\sim P_Z^r(f)}r_f^r(Z),
  \quad
  \widehat R^r(\rho)=\E_{f\sim\rho}\widehat L^r(f),
  \quad
  \widetilde R^o(\rho)=\E_{f\sim\rho}\E_{Z\sim P_Z^o(f)}\widetilde r_f^o(Z).
  \label{eq:main-risk-definitions}
\end{equation}
Here $R^r(\rho)$ is the trial-population risk, $\widehat R^r(\rho)$ is its empirical RCT analogue, and $\widetilde R^o(\rho)$ is the OS risk after subtracting the modeled baseline and comparative biases. To define the calibration residual, let $d_\ell(u,v)=\sup_t\left|\mathcal L_\ell(t,u)-\mathcal L_\ell(t,v)\right|$, and set
\begin{equation}
  \varepsilon_{\mathrm{cal}}(\rho)=
  \E_{f\sim \rho}\E_{Z\sim P_Z^r(f)}
  \sum_{a=0}^{1}\pi^r(a\mid Z)\,
  d_\ell\{\eta_\star^r(a,Z),\eta_{\mathrm{deb},f}^o(a,Z)\}
  \label{eq:main-calibration-terms}
\end{equation}
The residual $\varepsilon_{\mathrm{cal}}(\rho)$ vanishes when the debiased OS surface matches the RCT surface on the trial latent support. Because it depends on unknown source surfaces, the result below is best read as a target-risk decomposition and design guide rather than as a directly computable finite-sample guarantee.

\begin{theorem}[PAC-Bayes and IPM-calibration target-risk decomposition]
\label{thm:pacbayes}
Assume the loss is bounded in $[0,1]$ and that the IPM regularity condition stated in Appendix~\ref{app:thm-pacbayes} holds for $\rho$; informally, the debiased OS conditional risk must be controlled by the discriminator class $\calG$ up to a constant $L$. Define $\Delta_Z(\rho)=
  \E_{f\sim \rho}\ipm_{\calG}\{P_Z^r(f),P_Z^o(f)\}$. With the quantities in Eqs.~\eqref{eq:main-risk-definitions}--\eqref{eq:main-calibration-terms}, for any data-independent prior $\rho_0$ over $f$, with probability at least $1-\delta$ over the RCT sample, every such posterior $\rho$ satisfies
\begin{equation}
  R^r(\rho)
  \le
  \min\left[
  \widehat R^r(\rho)+
  \sqrt{\dfrac{\KL(\rho\|\rho_0)+\log(2\sqrt{n^r}/\delta)}{2n^r}},
  \,
  \widetilde R^o(\rho)+L\Delta_Z(\rho)+\varepsilon_{\mathrm{cal}}(\rho)
  \right].
  \label{eq:pac_bound_rct}
\end{equation}
\end{theorem}

The first term inside the minimum is the RCT-only PAC-Bayes certificate; the second is the OS IPM-calibration certificate. The bounded-loss assumption can be replaced by a uniform sub-Gaussian tail condition: for every draw $f$, the centered RCT loss $U_f=\ell\{\eta_f^r(A^r,Z^r),Y^r\}-\E\ell\{\eta_f^r(A^r,Z^r),Y^r\}$ satisfies $\E\{\exp(\lambda U_f)\}\le \exp(\lambda^2\sigma^2/2)$ for all real $\lambda$. This is a tail condition on the loss, not an assumption that the outcome itself is Gaussian; ordinary unbounded squared loss requires clipping, localization, or stronger moment conditions to satisfy it. Appendix~\ref{app:thm-pacbayes} gives the same $\mathcal{O}((n^r)^{-1/2})$ rate with a different constant under this condition.

This decomposition explains why alignment alone is insufficient. Small $\Delta_Z(\rho)$ only establishes that the two latent distributions are close; useful borrowing also requires the debiased OS surface to be calibrated to the trial surface on the shared support, and Theorem~\ref{thm:functional} limits how much treatment-contrast precision can be taken from the OS while maintaining that protection.

\section{Experiments}
\label{sec:experiments}

The experiments evaluate the four claims from Section~\ref{sec:intro}: (i) \emph{bias-limited borrowing} --- OS contrast information saturates at the ceiling predicted by Corollary~\ref{thm:scalar}; (ii) \emph{negative-transfer protection} --- borrowing rarely increases error relative to an RCT-only analysis; (iii) \emph{calibration} --- 90\% credible intervals stay close to nominal coverage; and (iv) \emph{competitive point accuracy}. Three studies address them: a synthetic factorial study in which the true CATE is known by construction (claims i--iv), a semi-synthetic benchmark modeled on the Tennessee STAR class-size trial (claims iii--iv), and a pediatric-obesity external-control application with real trial and EHR data (claims ii--iii). The main text gives each design and its main results; Appendix~\ref{app:exp-details} gives the full specifications: the data-generating equations and generation pseudocode for the synthetic study, the B-CALM estimator and selection procedure exactly as fit here, the procedure that constructs the semi-synthetic benchmark, the real-data cohort description, and the metric formulas.

\vspace{-0.4em}
\subsection{Synthetic factorial study}
\vspace{-0.4em}

\paragraph{Design.} A three-dimensional latent state generates trial covariates $X^r=(W,U)$ and observational covariates $X^o=(W,V)$ through source-specific noisy linear maps. This keeps the true CATE known even though the sources share only the block $W$. The factorial grid crosses three covariate-overlap regimes (\emph{shared-only}, \emph{balanced}, and \emph{mismatch-heavy}, with increasingly source-specific covariate loadings and increasingly nonlinear outcome surfaces), three OS bias regimes (none, baseline-only, comparative), and, on the balanced DGP, a sweep that scales the comparative-bias function by $\{0,0.2,0.5,1.0\}$, using $n^r\in\{250,500\}$, $n^o=5{,}000$, and $50$ replicates per cell. OS treatment assignment is confounded through the observational covariates in every cell. Appendix~\ref{app:dgp-synthetic} gives the generative equations, the assignment model, the full condition grid, and generation pseudocode (Algorithm~\ref{alg:dgp}).

\paragraph{B-CALM configuration.} All synthetic experiments use the closed-form Gaussian instance of the model in Section~\ref{sec:model}: encoders are fixed identity maps, so each source's covariate vector serves as its latent state; $\mu$, $\tau$, $b_0$, and $b_\Delta$ are represented on a prespecified polynomial--trigonometric basis; and the posterior over basis coefficients is Gaussian in closed form. We compute posterior CATE means, variances, and 90\% intervals analytically, without Monte Carlo sampling. No encoder, decoder, or alignment penalty is trained. The comparative-bias scale is selected within each replicate by empirical-Bayes model averaging over the fixed candidate set $\sigma_\Delta\in\{0.01,0.05,0.2,0.5,1.0,\infty\}$ with global fixed prior weights that are asymmetric by design (skeptical on the tight scales, a mild mode at moderate borrowing, and a heavy diffuse tail) and prior mean $m_\Delta=0$; Algorithm~\ref{alg:bcalm} in Appendix~\ref{app:bcalm-implementation} states the fitting and selection procedure. These experiments isolate bias-limited borrowing in a fixed common working representation; they do not evaluate learned latent alignment.

\paragraph{Baselines.} We compare with RCT-only and OS-only fits of the same Gaussian model, pooled RCT+OS (the same model with $b_0=b_\Delta\equiv0$), a causal forest (CF) fit to the pooled data \citep{wager2018estimation}, and the calibrated-borrowing family R-OSCAR, MR-OSCAR, and CALM \citep{asiaee2025roscar,pal2026mroscar,asiaee2026calm}. R-OSCAR calibrates OS outcome predictions built from the shared covariates only; MR-OSCAR first imputes the OS-only covariates for trial units and then calibrates; and CALM replaces explicit imputation with source-specific embeddings before calibration. In this Gaussian harness the shared trial-calibration stage dominates all three, and their results agree to within Monte Carlo noise, so we report them as a single \emph{OSCAR/CALM} entry. Bayesian causal forests (BCF) \citep{hahn2020bayesian} require a full MCMC run per fit, so we evaluate them in the real-data study, where the number of fits is small.

\paragraph{Metrics.} We score each replicate on $600$ fresh covariate draws from the trial distribution, where the true CATE is known: RMSE of the posterior-mean CATE, empirical coverage and mean width of the pointwise 90\% credible intervals, and negative transfer (NT), the fraction of replicates in which a method's RMSE exceeds the RCT-only RMSE on the same replicate. Formulas are in Appendix~\ref{app:metrics}.

\paragraph{Results.}
\begin{table}[t]
\centering
\small
\caption{Synthetic factorial study: metric averages over the full condition grid ($24$ cells $\times$ $50$ replicates; Appendix~\ref{app:dgp-synthetic}). Cov.\ and Width refer to pointwise 90\% credible intervals evaluated on trial-distribution covariates; NT is the negative-transfer rate relative to RCT-only. R-OSCAR, MR-OSCAR, and CALM coincide in this harness and are reported as one OSCAR/CALM row. Per-cell results: Appendix~\ref{app:main-table}, Table~\ref{tab:main}.}
\label{tab:main_summary}
\begin{tabular}{@{}lrrrr@{}}
\toprule
Method   & RMSE  & Cov. & Width & NT \\
\midrule
B-CALM            & 0.206 & 0.939 & 0.739 & 0.023 \\
OSCAR/CALM        & 0.228 & 0.890 & 0.742 & 0.012 \\
RCT-only          & 0.318 & 0.902 & 1.017 & 0.000 \\
pooled            & 0.305 & 0.581 & 0.299 & 0.292 \\
OS-only           & 0.193 & 0.677 & 0.345 & 0.192 \\
CF                & 0.308 & 0.811 & 0.764 & 0.299 \\
\bottomrule
\end{tabular}
\end{table}

Table~\ref{tab:main_summary} summarizes the four metrics over the factorial grid; per-cell results are in Appendix~\ref{app:main-table}. B-CALM is the only borrowing method whose average coverage reaches the nominal $0.90$ (0.939), while also improving accuracy: RMSE 0.206 against 0.318 for RCT-only, with near-zero negative transfer (0.023). The two point-accuracy leaders lose calibration: OS-only (RMSE 0.193) covers at 0.677, and pooled has the narrowest intervals (width 0.299) at 0.581 coverage and a 0.292 negative-transfer rate.

\begin{figure}[t]
\centering
  \includegraphics[width=\linewidth]{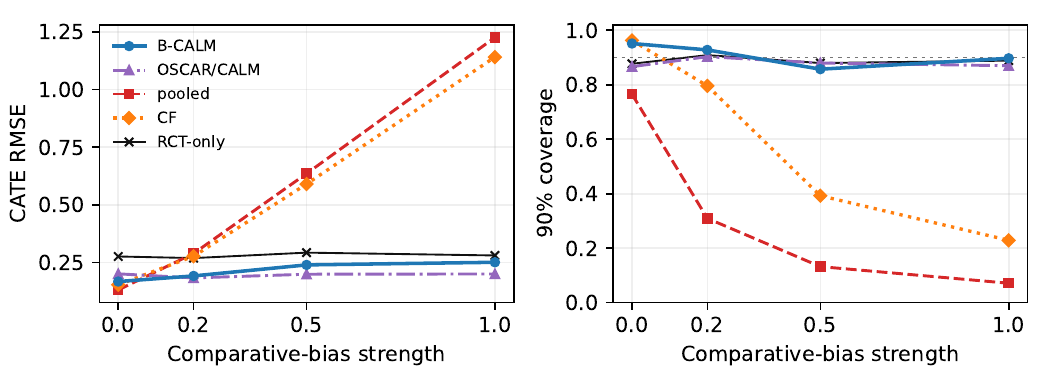}
  \caption{Comparative-bias sweep on the balanced DGP ($n^r=500$, $n^o=5{,}000$; each point averages 50 replicates). Left: CATE RMSE. Right: empirical coverage of pointwise 90\% credible intervals, with nominal coverage marked by the horizontal dashed line. R-OSCAR, MR-OSCAR, and CALM coincide in this harness and are plotted as the single OSCAR/CALM line; pooled and CF nearly overlap in the RMSE panel under strong bias.}
  \label{fig:bias_sweep}
\end{figure}

Figure~\ref{fig:bias_sweep} isolates the comparative-bias axis: the balanced DGP is fixed while the comparative-bias function is scaled up, holding baseline bias and OS confounding constant. The RMSE panel shows the accuracy tradeoff across methods. Pooled and CF are the most accurate when the OS contrast is unbiased, but their errors grow as the bias strengthens (pooled RMSE reaches 1.226 at the strongest setting). OSCAR/CALM is nearly flat across the sweep and beats B-CALM under strong bias (RMSE 0.199 versus 0.239 at strength 0.5): it declines to borrow contrast information, which protects it there but also means it gains nothing when the OS is clean (0.201 versus B-CALM's 0.167 at strength 0). The coverage panel separates B-CALM from the point-accuracy leaders: its coverage stays near nominal at every bias level while pooled falls to 0.07 and CF to 0.23 at the strongest setting. B-CALM's weakest cell is 0.857 coverage at strength 0.5, the transition region where the empirical-Bayes weights still place mass on borrowing scales (the exported posterior weight on the moderate scale $\sigma_\Delta=0.2$ averages 0.91 there, collapsing to 0.03 by strength 1.0); we return to this cost of adaptivity in the Discussion.

\begin{figure}[t]
\centering
  \includegraphics[width=\linewidth]{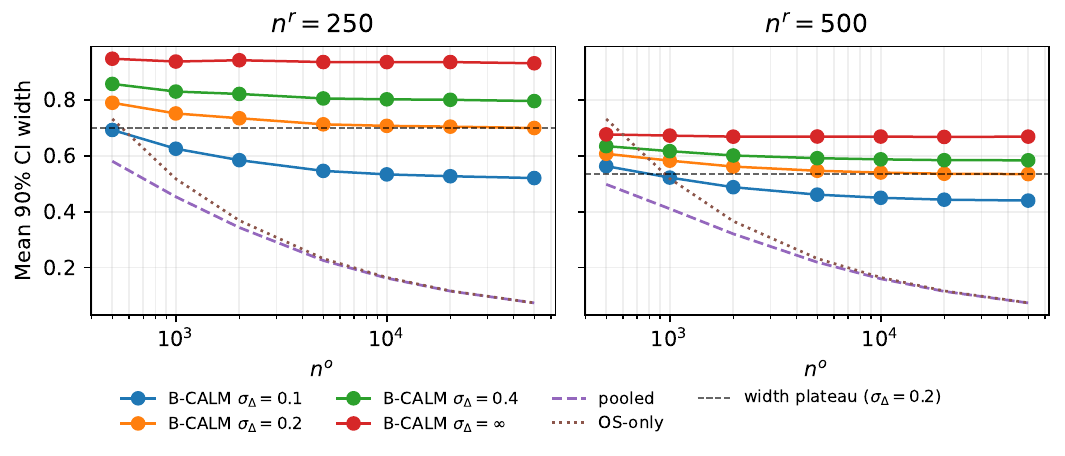}
  \caption{Interval-width saturation: mean 90\% credible-interval width as $n^o$ grows with $n^r$ fixed (log horizontal scale; 50 replicates per point). B-CALM interval widths flatten to a $\sigma_\Delta$-dependent floor (the horizontal dashed reference marks the $\sigma_\Delta=0.2$ plateau), while the pooled and OS-only curves, which nearly overlap, contract toward zero regardless of bias. Effective-sample-size values appear in Appendix~\ref{app:saturation-ess}, Table~\ref{tab:saturation}.}
  \label{fig:saturation}
\end{figure}

\paragraph{Observational sample-size domination.} This experiment fixes $n^r$ and grows $n^o$ from $500$ to $50{,}000$ under comparative bias, fitting B-CALM at fixed $\sigma_\Delta\in\{0.1,0.2,0.4,\infty\}$ with prior mean $m_\Delta=0.2$ (Appendix~\ref{app:reproducibility}). Pooled and OS-only intervals contract as $(n^o)^{-1/2}$ regardless of bias; B-CALM intervals contract until they reach the floor implied by $\sigma_\Delta^2$ and then saturate (Figure~\ref{fig:saturation}). At the largest $n^o$ the empirical effective sample size tracks the theoretical saturation value $\sigma_o^2/\sigma_\Delta^2$ (for example, 105.6 empirical versus 99.8 theoretical at $\sigma_\Delta=0.2$; Appendix~\ref{app:saturation-ess}, Table~\ref{tab:saturation}). The small nonzero empirical ESS at $\sigma_\Delta=\infty$ (about 4) comes from the finite numerical implementation of the diffuse prior and information flowing through the shared baseline surface, not comparative borrowing.

\paragraph{Ablation.} On the balanced DGP at an intermediate bias strength ($0.4$, chosen outside both the sweep grid and the $\sigma_\Delta$ candidate set), removing the comparative-bias function collapses coverage from $0.915$ to $0.324$ while RMSE moves far less (0.195 to 0.259): $b_\Delta$ is load-bearing for calibration, not point accuracy. Forcing $b_0\equiv0$ or replacing the posterior with covariance-scaled point-estimate surrogates also under-covers (Appendix~\ref{app:ablation-table}, Table~\ref{tab:ablation}).

Posterior calibration curves across nominal levels match the 90\% snapshot: B-CALM tracks the diagonal while pooled and CF under-cover (Appendix~\ref{app:calibration}, Figure~\ref{fig:calibration}).

\vspace{-0.4em}
\subsection{Semi-synthetic benchmark modeled on the STAR trial}
\vspace{-0.4em}

The Tennessee Student/Teacher Achievement Ratio (STAR) class-size trial is a standard covariate template for semi-synthetic causal benchmarks \citep{kallus2018removing}: realistic trial covariates and marginals are combined with known, imposed potential-outcome surfaces, so evaluation against the true CATE remains possible. Following the construction of \citet{asiaee2025roscar}, we generate a trial-like sample ($n^r=500$) and an observational sample ($n^o=5{,}000$) over seven STAR-type covariates (baseline score, sex, race, free-lunch status, small-class assignment, teacher experience, and school quality), with source-shifted marginals that mimic RCT-versus-EHR mismatch, confounded OS treatment assignment, and known baseline and comparative bias functions; the full routine is in Appendix~\ref{app:semi-construction}. Over $50$ replicates, B-CALM is statistically indistinguishable from the best RMSE (0.191 versus pooled's 0.194) while raising coverage from pooled's 0.425 to 0.936; MR-OSCAR, the family member run in this study, covers at 0.889 with RMSE 0.228 (Appendix~\ref{app:semi-synthetic}, Table~\ref{tab:semi}). Varying $\sigma_\Delta$ traces posterior CATEs at seven representative covariate profiles that interpolate between pooled ($\sigma_\Delta\to0$) and RCT-only ($\sigma_\Delta\to\infty$) inference (Appendix~\ref{app:semi-synthetic}, Figure~\ref{fig:semi-sensitivity}); this sensitivity curve is the report a B-CALM analysis would provide.

\subsection{Real-data application: external-control augmentation}
\label{sec:gps-real}
\vspace{-0.4em}

\paragraph{Data and identification.}
We evaluate B-CALM on a multi-site pediatric-obesity randomized trial ($n^r=385$) augmented with a cross-sectional EHR cohort ($n^o=15{,}150$) from the same health systems and treated as external controls. The outcome is 12-month weight-for-length z-score; the six shared covariates are site, sex, race/ethnicity, language, insurance, and baseline weight-for-length z-score, with five trial-only covariates retained for sensitivity analyses; Appendix~\ref{app:gps-details} describes the cohort, preprocessing, and model configuration. Because $A_j^o\equiv 0$ on every EHR row, the OS likelihood involves only $\mu(z)+b_0(z)$; the comparative-bias prior $\sigma_\Delta$ is unidentified and held diffuse, and the baseline-commensurability prior $\sigma_0$ on $b_0(z)$ controls borrowing strength.

\begin{table}[t]
\centering
\footnotesize
\setlength{\tabcolsep}{4pt}
\caption{Pediatric-obesity external-control augmentation. Width reduction is relative to RCT-only; Skep./w.\ is the bias-susceptibility ratio (bias risk implied by a skeptical $\sigma_0=0.6$ prior for each method's estimate, divided by that method's nominal width; values above one flag intervals narrower than the bias they may absorb). Full baselines, including BCF and OSCAR/CALM, are in Appendix~\ref{app:gps-real}, Table~\ref{tab:gps-main}.}
\label{tab:gps-main-compact}
\begin{tabular}{lccc}
\toprule
Method & ATE [90\% CI] & Width (Red.) & Skep./w. \\
\midrule
B-CALM   & $-0.184\,[-0.336,-0.031]$ & $0.305\,(9.4\%)$  & $0.34$ \\
RCT-only & $-0.174\,[-0.345,-0.008]$ & $0.337$           & $0.00$ \\
pooled   & $-0.191\,[-0.313,-0.074]$ & $0.239\,(28.9\%)$ & $1.05$ \\
CF       & $-0.162\,[-0.211,-0.112]$ & $0.099\,(70.6\%)$ & $2.91$ \\
\bottomrule
\end{tabular}
\end{table}

\vspace{-0.6em}
\paragraph{Main result.}
At the selected $\sigma_0=0.05$ (chosen from candidate baseline-bias prior scales by a prior-weighted score; Appendix~\ref{app:gps-details}), B-CALM estimates an ATE of $-0.184$ with 90\% credible interval $[-0.336,-0.031]$, a $9.4\%$ width reduction over RCT-only (Table~\ref{tab:gps-main-compact}). Naive borrowing baselines reach larger reductions by implicitly setting $b_0\equiv0$; the full table, including BCF \citep{hahn2020bayesian} and OSCAR-family baselines, is in Appendix~\ref{app:gps-real}, Table~\ref{tab:gps-main}. The $\sigma_0$ sweep in Appendix~\ref{app:gps-real}, Figure~\ref{fig:gps-real}, traces the borrowing pattern: $\sigma_0\to0$ approaches pooled, $\sigma_0\to\infty$ returns to RCT-only, and $\sigma_0=0.05$ sits at the bias-protected interior of the curve.

\vspace{-0.6em}
\paragraph{Bias-susceptibility ratio.}
The skeptical ratio in Table~\ref{tab:gps-main-compact} is the 90\% range of B-CALM-calibrated bias risk under a skeptical $\sigma_0=0.6$ prior (matching the default baseline-bias prior in the synthetic experiments and an order of magnitude above the selected $0.05$), divided by a method's nominal width; values above one mean the interval is narrower than the bias it would absorb. Pooled ($1.05$) and CF ($2.91$) achieve narrower intervals by accepting bias risk above their reported width, while B-CALM remains below one. Leave-one-covariate masking and $b_0$ recovery diagnostics are in Appendix~\ref{app:gps-real}.

\vspace{-0.6em}
\section{Discussion and conclusion}
\label{sec:discussion}
\vspace{-0.6em}
\paragraph{Discussion.}
The function-valued bias-limited information bound (Theorem~\ref{thm:functional}) and its scalar instance (Corollary~\ref{thm:scalar}) give a closed-form prediction that the experiments confirm: observational borrowing produces a $\sigma_\Delta$-dependent uncertainty floor instead of asymptotically zero credible-interval width. The same mechanism explains the one weak cell in the bias sweep: at intermediate comparative bias the empirical-Bayes weights retreat from borrowing gradually, and coverage dips to $0.857$ before recovering; fixing $\sigma_\Delta$ at a conservative value removes this dip at the price of wider intervals, as the sample-size domination experiment shows. In practice, this lets an analyst report a sensitivity analysis over prespecified values of $\sigma_\Delta$, rather than committing to one borrowing weight, and compute the $\ess^o$ implied by each value. The pediatric-obesity application illustrates the same pattern with one-arm OS data: B-CALM's modest $9.4\%$ width reduction at $\sigma_0=0.05$ is what is left after the model explicitly accounts for plausible EHR-vs-trial control bias, while the larger reductions reported by pooled, CF, and BCF rest on the implicit and unverifiable assumption that $b_0\equiv 0$.

\vspace{-0.6em}
\paragraph{Practical interpretation.}
The practical output of B-CALM is not a single borrowed estimate, but a calibrated borrowing analysis: a posterior for the RCT-defined estimand, a sensitivity parameter for OS contrast bias, and an ESS diagnostic describing how much information the OS is allowed to contribute. The reported estimand changes accordingly. Rather than asking whether the OS is ``similar enough'' to pool, the analyst reports how conclusions move as $\sigma_\Delta$ or $\sigma_0$ varies, and whether the resulting interval contraction is commensurate with the bias the model is willing to absorb. B-CALM is most useful when external data are potentially informative but not trustworthy enough to be treated as randomized evidence.

\vspace{-0.6em}
\paragraph{Limitations.}
(i) $b_\Delta$ is not identifiable without prior restrictions; reporting a sensitivity curve over prespecified $\sigma_\Delta$ values makes this visible rather than committing to one value. (ii) Latent alignment is only as good as the shared covariates: with little overlap, B-CALM reverts toward RCT-only — the safe failure mode, but observational data of poor covariate match adds no precision. (iii) Posterior computation is expensive at EHR scale; sparse GPs, SVI, and amortized encoders are practical remedies. Future directions are in Appendix~\ref{app:future-directions}.

{\small
\bibliographystyle{plainnat}
\bibliography{references}

@article{asiaee2025roscar,
  title = {Improving Precision of {RCT}-Based {CATE} Estimation using Data Borrowing with Double Calibration},
  author = {Asiaee, Amir and Di Gravio, Chiara and Beck, Cole and Mei, Yuting and Pal, Samhita and Huling, Jared D.},
  journal = {arXiv preprint arXiv:2306.17478},
  year = {2025},
  eprint = {2306.17478},
  archivePrefix = {arXiv},
  primaryClass = {stat.ME}
}

@article{pal2026mroscar,
  title = {Improving {RCT}-Based {CATE} Estimation Under Covariate Mismatch via Double Calibration},
  author = {Pal, Samhita and Huling, Jared D. and Asiaee, Amir},
  journal = {arXiv preprint arXiv:2603.17066},
  year = {2026},
  eprint = {2603.17066},
  archivePrefix = {arXiv},
  primaryClass = {stat.ME}
}

@article{asiaee2026calm,
  title = {Improving {RCT}-Based Treatment Effect Estimation Under Covariate Mismatch via Calibrated Alignment},
  author = {Asiaee, Amir and Pal, Samhita},
  journal = {arXiv preprint arXiv:2603.19186},
  year = {2026},
  eprint = {2603.19186},
  archivePrefix = {arXiv},
  primaryClass = {cs.LG}
}

@article{asiaee2026sharp,
  title = {Sharp Bounds for Treatment Effect Generalization under Outcome Distribution Shift},
  author = {Asiaee, Amir and Pal, Samhita and Beck, Cole and Huling, Jared D.},
  journal = {arXiv preprint arXiv:2602.09595},
  year = {2026},
  note = {To appear, 5th Conference on Causal Learning and Reasoning (CLeaR)},
  eprint = {2602.09595},
  archivePrefix = {arXiv}
}

@article{asiaee2026ovb,
  title = {Omitted-Variable Sensitivity Analysis for Generalizing Randomized Trials},
  author = {Asiaee, Amir and Pal, Samhita and Huling, Jared D.},
  journal = {arXiv preprint arXiv:2603.27788},
  year = {2026},
  eprint = {2603.27788},
  archivePrefix = {arXiv},
  primaryClass = {stat.ME}
}

@inproceedings{bica2022heterogeneous,
  title = {Transfer Learning on Heterogeneous Feature Spaces for Treatment Effects Estimation},
  author = {Bica, Ioana and van der Schaar, Mihaela},
  booktitle = {Advances in Neural Information Processing Systems},
  volume = {35},
  pages = {37184--37198},
  year = {2022}
}

@article{dimitriou2024causalicm,
  title = {{Causal-ICM}: A Data Fusion Framework for Heterogeneous Treatment Effect Estimation with Multi-Task Gaussian Processes},
  author = {Dimitriou, Evangelos and Fong, Edwin and Magelund Tarp, Jens and Diaz-Ordaz, Karla and Lehmann, Brieuc},
  journal = {arXiv preprint arXiv:2405.20957},
  year = {2024},
  eprint = {2405.20957},
  archivePrefix = {arXiv},
  primaryClass = {stat.ME}
}

@inproceedings{alaa2017bayesian,
  title = {Bayesian Inference of Individualized Treatment Effects using Multi-task Gaussian Processes},
  author = {Alaa, Ahmed M. and van der Schaar, Mihaela},
  booktitle = {Advances in Neural Information Processing Systems},
  volume = {30},
  pages = {3424--3432},
  year = {2017}
}

@article{hahn2020bayesian,
  title = {Bayesian Regression Tree Models for Causal Inference: Regularization, Confounding, and Heterogeneous Effects (with Discussion)},
  author = {Hahn, P. Richard and Murray, Jared S. and Carvalho, Carlos M.},
  journal = {Bayesian Analysis},
  volume = {15},
  number = {3},
  pages = {965--1056},
  year = {2020},
  doi = {10.1214/19-BA1195}
}

@article{ibrahim2000power,
  title = {Power Prior Distributions for Regression Models},
  author = {Chen, Ming-Hui and Ibrahim, Joseph G.},
  journal = {Statistical Science},
  volume = {15},
  number = {1},
  pages = {46--60},
  year = {2000},
  doi = {10.1214/ss/1009212673}
}

@article{hobbs2011hierarchical,
  title = {Hierarchical Commensurate and Power Prior Models for Adaptive Incorporation of Historical Information in Clinical Trials},
  author = {Hobbs, Brian P. and Carlin, Bradley P. and Mandrekar, Sumithra J. and Sargent, Daniel J.},
  journal = {Biometrics},
  volume = {67},
  number = {3},
  pages = {1047--1056},
  year = {2011},
  doi = {10.1111/j.1541-0420.2011.01564.x}
}

@article{hobbs2012commensurate,
  title = {Commensurate Priors for Incorporating Historical Information in Clinical Trials Using General and Generalized Linear Models},
  author = {Hobbs, Brian P. and Sargent, Daniel J. and Carlin, Bradley P.},
  journal = {Bayesian Analysis},
  volume = {7},
  number = {3},
  pages = {639--674},
  year = {2012},
  doi = {10.1214/12-BA722}
}

@article{schmidli2014robust,
  title = {Robust Meta-Analytic-Predictive Priors in Clinical Trials with Historical Control Information},
  author = {Schmidli, Heinz and Gsteiger, Sandro and Roychoudhury, Satrajit and O'Hagan, Anthony and Spiegelhalter, David and Neuenschwander, Beat},
  journal = {Biometrics},
  volume = {70},
  number = {4},
  pages = {1023--1032},
  year = {2014},
  doi = {10.1111/biom.12242}
}

@misc{harrell2023hxcontrol,
  title = {Incorporating Historical Control Data Into an {RCT}},
  author = {Harrell, Frank E.},
  year = {2023},
  howpublished = {Statistical Thinking blog},
  note = {Accessed 2026-04-25}
}

@article{degtiar2023review,
  title = {A Review of Generalizability and Transportability},
  author = {Degtiar, Irina and Rose, Sherri},
  journal = {Annual Review of Statistics and Its Application},
  volume = {10},
  number = {1},
  pages = {501--524},
  year = {2023},
  publisher = {Annual Reviews}
}

@article{colnet2024causal,
  title = {Causal Inference Methods for Combining Randomized Trials and Observational Studies: A Review},
  author = {Colnet, B{\'e}n{\'e}dicte and Mayer, Imke and Chen, Guanhua and Dieng, Awa and Li, Ruohong and Varoquaux, Ga{\"e}l and Vert, Jean-Philippe and Josse, Julie and Yang, Shu},
  journal = {Statistical Science},
  volume = {39},
  number = {1},
  pages = {165--191},
  year = {2024},
  doi = {10.1214/23-STS889}
}

@article{gretton2012kernel,
  title = {A Kernel Two-Sample Test},
  author = {Gretton, Arthur and Borgwardt, Karsten M. and Rasch, Malte J. and Sch\"olkopf, Bernhard and Smola, Alexander},
  journal = {Journal of Machine Learning Research},
  volume = {13},
  pages = {723--773},
  year = {2012}
}

@article{wager2018estimation,
  title = {Estimation and Inference of Heterogeneous Treatment Effects using Random Forests},
  author = {Wager, Stefan and Athey, Susan},
  journal = {Journal of the American Statistical Association},
  volume = {113},
  number = {523},
  pages = {1228--1242},
  year = {2018},
  doi = {10.1080/01621459.2017.1319839}
}

@inproceedings{shalit2017estimating,
  title = {Estimating Individual Treatment Effect: Generalization Bounds and Algorithms},
  author = {Shalit, Uri and Johansson, Fredrik D. and Sontag, David},
  booktitle = {Proceedings of the 34th International Conference on Machine Learning},
  series = {Proceedings of Machine Learning Research},
  volume = {70},
  pages = {3076--3085},
  year = {2017}
}

@inproceedings{johansson2016learning,
  title = {Learning Representations for Counterfactual Inference},
  author = {Johansson, Fredrik D. and Shalit, Uri and Sontag, David},
  booktitle = {Proceedings of the 33rd International Conference on Machine Learning},
  series = {Proceedings of Machine Learning Research},
  volume = {48},
  pages = {3020--3029},
  year = {2016}
}

@inproceedings{kallus2018removing,
  title = {Removing Hidden Confounding by Experimental Grounding},
  author = {Kallus, Nathan and Puli, Aahlad Manas and Shalit, Uri},
  booktitle = {Advances in Neural Information Processing Systems},
  volume = {31},
  year = {2018}
}

@inproceedings{hatt2022combining,
  title = {Combining Observational and Randomized Data for Estimating Heterogeneous Treatment Effects},
  author = {Hatt, Tobias and Berrevoets, Jeroen and Curth, Alicia and Feuerriegel, Stefan and van der Schaar, Mihaela},
  booktitle = {ICML 2022 Workshop on Spurious Correlations, Invariance and Stability},
  year = {2022},
  eprint = {2202.12891},
  archivePrefix = {arXiv},
  primaryClass = {stat.ML}
}

@inproceedings{ilse2023combining,
  title = {Combining Interventional and Observational Data Using Causal Reductions},
  author = {Ilse, Maximilian and Forr{\'e}, Patrick and Welling, Max and Mooij, Joris M.},
  booktitle = {Transactions on Machine Learning Research},
  year = {2023},
  eprint = {2103.04786},
  archivePrefix = {arXiv},
  primaryClass = {stat.ML}
}

@article{kaizer2018bayesian,
  title = {Bayesian Hierarchical Modeling Based on Multisource Exchangeability},
  author = {Kaizer, Alexander M. and Koopmeiners, Joseph S. and Hobbs, Brian P.},
  journal = {Biostatistics},
  volume = {19},
  number = {2},
  pages = {169--184},
  year = {2018},
  doi = {10.1093/biostatistics/kxx031}
}

@article{neuenschwander2010summarizing,
  title = {Summarizing Historical Information on Controls in Clinical Trials},
  author = {Neuenschwander, Beat and Capkun-Niggli, Gorana and Branson, Michael and Spiegelhalter, David J.},
  journal = {Clinical Trials},
  volume = {7},
  number = {1},
  pages = {5--18},
  year = {2010},
  doi = {10.1177/1740774509356002}
}

@article{vanrosmalen2018including,
  title = {Including Historical Data in the Analysis of Clinical Trials: Is It Worth the Effort?},
  author = {van Rosmalen, Joost and Dejardin, David and van Norden, Yvette and L{\"o}wenberg, Bob and Lesaffre, Emmanuel},
  journal = {Statistical Methods in Medical Research},
  volume = {27},
  number = {10},
  pages = {3167--3182},
  year = {2018},
  doi = {10.1177/0962280217694506}
}

@article{kotalik2021framework,
  title = {Dynamic Borrowing in the Presence of Treatment Effect Heterogeneity},
  author = {Kotalik, Ales and Vock, David M. and Donny, Eric C. and Hatsukami, Dorothy K. and Koopmeiners, Joseph S.},
  journal = {Biostatistics},
  volume = {22},
  number = {4},
  pages = {789--804},
  year = {2021},
  doi = {10.1093/biostatistics/kxz066}
}

@article{lewis2019borrowing,
  title = {Borrowing from Historical Control Data in Cancer Drug Development: A Cautionary Tale and Practical Guidelines},
  author = {Lewis, Connor Jo and Sarkar, Somnath and Zhu, Jiawen and Carlin, Bradley P.},
  journal = {Statistics in Biopharmaceutical Research},
  volume = {11},
  number = {1},
  pages = {67--78},
  year = {2019},
  doi = {10.1080/19466315.2018.1497533}
}
}

\appendix

\section{Additional derivations}
\label{app:functional-pac}

This appendix separates the algebraic borrowing bounds from the representation-risk decomposition. Appendix~\ref{app:scalar-finite-basis-proofs} proves the scalar and finite-basis information results used in the main text. Appendices~\ref{app:prop-limits-gp} and~\ref{app:encoder-aware-ess} give GP-limit and encoder-uncertainty extensions. Appendix~\ref{app:thm-pacbayes} contains the PAC-Bayes and IPM-calibration details for Theorem~\ref{thm:pacbayes}.

\subsection{Scalar and finite-basis information proofs}
\label{app:scalar-finite-basis-proofs}

\subsubsection{Proof of Corollary~\ref{thm:scalar} and posterior mean}
\label{app:thm-scalar-proof}

Integrating over $b$ gives $\widehat{\bar\tau}^o\mid\bar\tau\sim\Normal(\bar\tau,\sigma_o^2/n^o+\sigma_b^2)$. The posterior precision for a normal mean with independent normal observations and prior is the sum of prior and likelihood precisions, hence
\begin{equation}
  V_n^{-1}=v_0^{-1}+\frac{n^r}{\sigma_r^2}+\frac{1}{\sigma_o^2/n^o+\sigma_b^2}.
\end{equation}
Equating the OS precision $(\sigma_o^2/n^o+\sigma_b^2)^{-1}$ with $\ess^o/\sigma_o^2$ gives $\ess^o=\sigma_o^2/(\sigma_o^2/n^o+\sigma_b^2)=n^o/(1+n^o\sigma_b^2/\sigma_o^2)$, and $\lim_{n^o\to\infty}\ess^o=\sigma_o^2/\sigma_b^2$.

The corresponding posterior mean is
\begin{equation}
  m_n=V_n\left\{\frac{m_0}{v_0}+\frac{n^r\widehat{\bar\tau}^r}{\sigma_r^2}+\frac{\widehat{\bar\tau}^o}{\sigma_o^2/n^o+\sigma_b^2}\right\},
\end{equation}
which makes the borrowing weight explicit: the observational point estimate can be extremely precise for $\bar\tau+b$, but its weight for $\bar\tau$ is capped by $1/\sigma_b^2$.

\subsubsection{Proof of Theorem~\ref{thm:functional}}

The marginal covariance of the OS contrast vector is $\Omega=\sigma^2I+\Phi^o\Sigma_b(\Phi^o)^\top$. The Fisher information for $\beta_\tau$ is $I^o=(\Phi^o)^\top\Omega^{-1}\Phi^o$. If $\Phi^o$ has full column rank, the GLS information identity yields
\begin{equation}
  (\Phi^o)^\top(\sigma^2I+\Phi^o\Sigma_b(\Phi^o)^\top)^{-1}\Phi^o
  =\left\{\Sigma_b+\sigma^2((\Phi^o)^\top\Phi^o)^{-1}\right\}^{-1}.
\end{equation}
Since $\sigma^2((\Phi^o)^\top\Phi^o)^{-1}$ is positive semidefinite, $\Sigma_b^{-1}-I^o$ is also positive semidefinite. If the empirical feature covariance grows so that $((\Phi^o)^\top\Phi^o)^{-1}\to0$, then the information converges to $\Sigma_b^{-1}$.

\subsection{Gaussian-process limit in Proposition~\ref{prop:limits}}
\label{app:prop-limits-gp}

We state the Gaussian-process version in a form that separates two limits. The first is the flat comparative-bias limit, in which the OS contrast is absorbed by an increasingly diffuse prior on $b_\Delta$. The second is the proper GP discretization limit, in which a finite mesh approximation converges to the joint GP model with fixed $\sigma_\Delta^2<\infty$.

\paragraph{Proof sketch.}
For fixed $z$, the OS treated-control contrast is a noisy observation of $\tau(z)+b_\Delta(z)$. When $b_\Delta(z)\equiv 0$, this is a noisy observation of $\tau(z)$. When $b_\Delta(z)$ has an improper flat prior independent of $\tau(z)$, the transformation $(\tau,b_\Delta)\mapsto(\tau,c=\tau+b_\Delta)$ yields an OS likelihood depending on $c$ but not on $\tau$, so integrating over $c$ leaves the marginal posterior for $\tau$ proportional to the RCT likelihood times the prior, up to terms involving shared baseline parameters. The GP version below requires a uniform flat-precision condition on the OS contrast directions: if $C_m$ is the covariance induced by the mesh approximation to $b_\Delta$, then $\|(\sigma_{\Delta,m}^2 C_m)^{-1}\|_{\mathrm{op}}\to 0$ yields $L^1$ convergence of the OS marginal-likelihood ratio to one under the RCT-only posterior, hence total-variation convergence. Pointwise variance inflation is not sufficient when the mesh dimension grows, because high-frequency directions can retain nonvanishing precision (Remark~\ref{rem:pointwise-flatness-counterexample}).

\begin{proposition}[Flat GP comparative-bias limit]
\label{prop:gp-flat-limit}
Let $\calZ_0\subset\R^d$ be compact. Conditional on fixed encoders and on baseline terms, suppose the Gaussian contrast likelihoods can be written as
\begin{align}
  Y^r &= H^r\tau+\varepsilon^r, &
  \varepsilon^r&\sim\Normal(0,\Omega^r),\\
  Y^o &= H^o\tau+B_m b_m+\varepsilon^o, &
  \varepsilon^o&\sim\Normal(0,\Omega^o),
\end{align}
where $\Omega^r$ and $\Omega^o$ are positive definite, $H^r$ and $H^o$ are finite collections of bounded linear evaluation or quadrature functionals on $C(\calZ_0)$, and $b_m=(b_\Delta(z_1),\ldots,b_\Delta(z_m))^\top$ is evaluated on a mesh $Z_m=\{z_1,\ldots,z_m\}\subset\calZ_0$. Assume $b_m\sim\Normal(m_m,\sigma_{\Delta,m}^2K_b^{(m)})$, where $K_b^{(m)}=k_b(Z_m,Z_m)$ and $k_b$ is continuous, bounded, and strictly positive definite on $\calZ_0$. Let
\begin{equation}
  C_m=B_mK_b^{(m)}B_m^\top .
\end{equation}
Assume that $C_m$ is positive definite on the OS contrast subspace and that the flat-precision condition
\begin{equation}
  \left\|(\sigma_{\Delta,m}^2C_m)^{-1}\right\|_{\mathrm{op}}\longrightarrow0
  \label{eq:flat-precision-condition}
\end{equation}
holds. If the prior $\Pi_\tau$ for $\tau$ is a proper Gaussian measure on $C(\calZ_0)$, independent of $b_m$, then the marginal posterior $\Pi_m(\tau\mid Y^r,Y^o)$ converges in total variation to the RCT-only posterior
\begin{equation}
  \Pi^r(d\tau\mid Y^r)
  \propto
  \varphi_{\Omega^r}\{Y^r-H^r\tau\}\,\Pi_\tau(d\tau).
\end{equation}
\end{proposition}

\begin{proof}
After integrating $b_m$, the OS marginal likelihood for $\tau$ is
\begin{equation}
  L_m^o(\tau)=
  \varphi_{\Omega^o+\sigma_{\Delta,m}^2C_m}
  \{Y^o-H^o\tau-B_mm_m\}.
\end{equation}
Let $V_m=\Omega^o+\sigma_{\Delta,m}^2C_m$ and define the likelihood ratio
\begin{equation}
  r_m(\tau)=\frac{L_m^o(\tau)}{L_m^o(0)} .
\end{equation}
Because $\Omega^o$ is positive definite and Eq.~\eqref{eq:flat-precision-condition} holds, $\|V_m^{-1}\|_{\mathrm{op}}\to0$. For every fixed $\tau$ for which $H^o\tau$ is finite,
\begin{equation}
  \log r_m(\tau)
  =
  (H^o\tau)^\top V_m^{-1}(Y^o-B_mm_m)
  -\frac{1}{2}(H^o\tau)^\top V_m^{-1}(H^o\tau)
  \longrightarrow0 .
\end{equation}
The vector $H^o\tau$ is Gaussian under the proper GP prior $\Pi_\tau$, since $H^o$ contains finitely many bounded linear functionals. Gaussian exponential moments imply that, for all large $m$, $r_m(\tau)$ is dominated by an integrable function under the RCT-only posterior $\Pi^r(\cdot\mid Y^r)$. Hence $r_m\to1$ in $L^1\{\Pi^r(\cdot\mid Y^r)\}$ by dominated convergence.

The marginal posterior with the OS likelihood included has density
\begin{equation}
  \frac{d\Pi_m}{d\Pi^r}(\tau)
  =
  \frac{r_m(\tau)}
       {\int r_m(u)\,\Pi^r(du\mid Y^r)} .
\end{equation}
Since $r_m\to1$ in $L^1(\Pi^r)$, the denominator converges to one and
\begin{equation}
  \left\|\Pi_m-\Pi^r\right\|_{\mathrm{TV}}
  =
  \frac{1}{2}
  \int
  \left|
    \frac{r_m(\tau)}
       {\int r_m(u)\,\Pi^r(du\mid Y^r)}
    -1
  \right|
  \Pi^r(d\tau\mid Y^r)
  \longrightarrow0 .
\end{equation}
The same argument applies after integrating proper Gaussian nuisance parameters, such as $\mu$ and $b_0$, because Fubini's theorem applies to the joint Gaussian prior and the Gaussian likelihood is nonnegative and finite.
\end{proof}

\begin{remark}[Kernel and encoder conditions]
\label{rem:kernel-encoder-conditions}
The preceding proof uses only finite-dimensional covariance matrices, but the following standard conditions are sufficient for the mesh construction to be a discretization of a GP on the latent support. The support $\calZ_0$ is compact; the posterior encoder distributions for $Z^r$ and $Z^o$ are supported on $\calZ_0$ and have densities bounded above with respect to Lebesgue measure; the mesh fill distance tends to zero; and $k_b$ is continuous, bounded, and admits a Mercer expansion
\begin{equation}
  k_b(z,z')=\sum_{j=1}^{\infty}\lambda_je_j(z)e_j(z'),
\end{equation}
with $\lambda_j>0$, $\sum_j\lambda_j<\infty$, and uniformly bounded continuous eigenfunctions. These assumptions make the GP a tight Borel probability measure on $C(\calZ_0)$ when the sample paths are continuous, and they give convergence of finite collections of evaluation or stable quadrature functionals. The density bound on the encoder distributions is convenient for replacing empirical mesh sums by integrals over $P_Z^r$ and $P_Z^o$; it can be weakened to weak convergence of those quadrature rules against bounded continuous functions.
\end{remark}

\begin{remark}[Why pointwise flatness is not enough]
\label{rem:pointwise-flatness-counterexample}
Condition~\eqref{eq:flat-precision-condition} is a uniform flat-precision condition on the OS contrast directions. Pointwise variance inflation does not suffice when the mesh dimension grows. To see this, take a diagonal Mercer representation with eigenvalues $\lambda_j=j^{-2}$ and let the $m$th discretization retain the first $m$ basis coefficients. Set $\sigma_{\Delta,m}^2=m$. For each fixed coordinate $j$, $\sigma_{\Delta,m}^2\lambda_j\to\infty$, so the prior is pointwise diffuse on every fixed finite coordinate. Along the new coordinate $j=m$, however, the variance is $\sigma_{\Delta,m}^2\lambda_m=1/m\to0$. An OS contrast observation in that coordinate has marginal variance $\sigma_o^2+1/m$, not an exploding variance, and therefore supplies nonvanishing information about the corresponding coefficient of $\tau$. Thus the posterior cannot converge uniformly to the RCT-only posterior over the growing cylinders. Requiring $\sigma_{\Delta,m}^2\lambda_{\min}(C_m)\to\infty$, or equivalently Eq.~\eqref{eq:flat-precision-condition}, rules out this high-frequency escape.
\end{remark}

\begin{proposition}[Proper GP mesh limit]
\label{prop:gp-proper-limit}
Keep the setup of Proposition~\ref{prop:gp-flat-limit}, but fix $0<\sigma_\Delta^2<\infty$. Let $b_\Delta\sim\mathcal{GP}(m_\Delta,\sigma_\Delta^2k_b)$ on $C(\calZ_0)$, and suppose the mesh approximation satisfies
\begin{equation}
  B_mm_m\to H_bm_\Delta,
  \qquad
  B_mK_b^{(m)}B_m^\top\to H_bK_bH_b^\top
  \label{eq:mesh-covariance-convergence}
\end{equation}
in Euclidean norm for the finite OS functionals $H_b$. Then the marginal posterior for $\tau$ under the mesh model converges in total variation to the posterior under the joint GP model in which
\begin{equation}
  Y^o\mid\tau
  \sim
  \Normal\left\{
    H^o\tau+H_bm_\Delta,\,
    \Omega^o+\sigma_\Delta^2H_bK_bH_b^\top
  \right\}.
\end{equation}
Consequently, with a proper comparative-bias GP prior, the mesh limit is not the RCT-only posterior. It is the bias-limited posterior that treats the OS contrast as information about $\tau+b_\Delta$ and caps its information about $\tau$ through the covariance $\sigma_\Delta^2k_b$.
\end{proposition}

\begin{proof}
Let $L_m^{o,\mathrm{pr}}(\tau)$ be the proper-prior OS marginal likelihood obtained by integrating $b_m$. By Eq.~\eqref{eq:mesh-covariance-convergence}, the mean vector and covariance matrix in $L_m^{o,\mathrm{pr}}(\tau)$ converge to the corresponding mean vector and covariance matrix of the exact GP marginal likelihood $(L^o)^{\mathrm{GP}}(\tau)$ for every fixed $\tau$. The limiting covariance is positive definite because $\Omega^o$ is positive definite. The Gaussian density is continuous in its mean and covariance on positive-definite covariance matrices, and the same Gaussian exponential-moment domination used in Proposition~\ref{prop:gp-flat-limit} yields
\begin{equation}
  L_m^{o,\mathrm{pr}}(\tau)\to (L^o)^{\mathrm{GP}}(\tau)
  \quad\text{in }L^1\{\Pi^r(\cdot\mid Y^r)\}.
\end{equation}
Writing each posterior density with respect to $\Pi^r(\cdot\mid Y^r)$ and normalizing gives total-variation convergence. The displayed exact-GP likelihood is precisely the likelihood obtained by integrating the proper GP $b_\Delta$, so all OS contrast information enters through the inflated covariance $\Omega^o+\sigma_\Delta^2H_bK_bH_b^\top$.
\end{proof}

\begin{remark}[Lifting over encoder uncertainty]
\label{rem:lifting-over-encoder-uncertainty}
The proposition statements bound the conditional posterior given a fixed encoder draw $\phi=(\phi^r,\phi^o)$. Let $\Pi_m^\phi$ and $\Pi^{\phi,r}$ denote the corresponding conditional posteriors for $\tau$, and let $q(d\phi)$ be the encoder posterior. If $q(d\phi)$ is independent of the comparative-bias prior structure used in the propositions, and $\|\Pi_m^\phi-\Pi^{\phi,r}\|_{\mathrm{TV}}\to0$ holds uniformly in $\phi$ or with a $q$-integrable dominating function, then Fubini-Tonelli gives convergence of the marginal posterior over $(\tau,\phi)$ to the RCT-only marginal posterior over $(\tau,\phi)$:
\begin{equation}
  \left\|
    \Pi_m^\phi(d\tau)q(d\phi)
    -
    \Pi^{\phi,r}(d\tau)q(d\phi)
  \right\|_{\mathrm{TV}}
  \le
  \int
  \|\Pi_m^\phi-\Pi^{\phi,r}\|_{\mathrm{TV}}
  q(d\phi)
  \longrightarrow0 .
\end{equation}
If uniformity fails but pointwise convergence and a $q$-integrable bound on $\|\Pi_m^\phi-\Pi^{\phi,r}\|_{\mathrm{TV}}$ hold, dominated convergence gives the same total-variation statement.

One concrete sufficient condition for the dominating bound used in the conditional Gaussian likelihood-ratio proof is that $q(d\phi)$ places mass on encoders with uniformly bounded operator norms $\|H^{\phi,r}\|$ and $\|H^{\phi,o}\|$, uniformly bounded image $\|H^{\phi,o}\tau\|$ over the support of $\Pi_\tau$, and covariance eigenvalue bounds in the flat or proper mesh argument that are uniform in $\phi$.
\end{remark}

\subsection{Proof of Proposition~\ref{prop:encoder-aware-ess}}
\label{app:encoder-aware-ess}

\begin{proof}
Integrating over the independent Gaussian variables $b$ and $\xi(\phi^o)$ gives
\begin{equation}
  \widehat{\bar\tau}^o\mid\bar\tau
  \sim
  \Normal\left(
    \bar\tau,\,
    \frac{\sigma_o^2}{n^o}+\sigma_b^2+\sigma_{\mathrm{enc}}^2
  \right).
\end{equation}
The RCT contrast, the OS contrast, and the normal prior for $\bar\tau$ are conditionally independent normal contributions to the same mean parameter, so posterior precision is prior precision plus likelihood precisions:
\begin{equation}
  V_n^{-1}
  =
  v_0^{-1}
  +\frac{n^r}{\sigma_r^2}
  +\frac{1}{\sigma_o^2/n^o+\sigma_b^2+\sigma_{\mathrm{enc}}^2}.
\end{equation}
Equating the OS precision with ${\ess}^{o}_{\mathrm{enc}}/\sigma_o^2$ gives the displayed effective sample size, and taking $n^o\to\infty$ gives the limit.
\end{proof}

\begin{remark}[Estimating the encoder-aware ESS]
\label{rem:estimating-encoder-aware-ess}
In a B-CALM fit, $\widehat\sigma_{\mathrm{enc}}^2$ can be estimated as the variance of $\widehat{\bar\tau}^o$ contributed by posterior draws of $\phi^o$, with the observed data, bias draw, and outcome-surface draws held fixed. Plugging this estimate, together with estimates of $\sigma_o^2$ and $\sigma_b^2$, into Proposition~\ref{prop:encoder-aware-ess} gives a single empirical diagnostic $\widehat{\ess}^{o}_{\mathrm{enc}}$ for borrowing health in the scalar linear-Gaussian setting.
\end{remark}

\subsection{PAC-Bayes and IPM-calibration details}
\label{app:thm-pacbayes}

\subsubsection{Definitions and calibration residual}

We first make the calibration residual explicit, using the notation of Eqs.~\eqref{eq:main-risk-definitions}--\eqref{eq:main-calibration-terms}. For a full B-CALM draw $f=(\phi^r,\phi^o,\mu,\tau,b_0,b_\Delta)$, let $P_Z^r(f)$ and $P_Z^o(f)$ be the distributions of the latent variables induced by the trial and observational covariate distributions through $q_{\phi^r}$ and $q_{\phi^o}$. Write $\eta_f^r(a,z)=\mu_f(z)+a\tau_f(z)$ and $\eta_{\mathrm{deb},f}^o(a,z)=\eta_\star^o(a,z)-b_{0,f}(z)-a b_{\Delta,f}(z)$. Let $\pi^r(a\mid z)$ be the target trial treatment distribution used to define the risk, for instance the randomized trial distribution or the uniform interventional distribution over $a\in\{0,1\}$.

Let $\eta_\star^r(a,z)$ and $\eta_\star^o(a,z)$ denote the true source-specific conditional outcome-mean surfaces. Write $\mathcal L_\ell(t,m)=\E_{Y\sim P_m}\ell(t,Y)$. For a loss $\ell(t,y)$ and two outcome-mean parameters $u$ and $v$, define the loss discrepancy
\begin{equation}
  d_\ell(u,v)
  =
  \sup_{t}
  \left|
    \mathcal L_\ell(t,u)
    -
    \mathcal L_\ell(t,v)
  \right|,
  \label{eq:loss-discrepancy}
\end{equation}
where $P_m$ is the conditional outcome distribution indexed by $m$. In Gaussian squared-error models with clipped predictions or in Bernoulli log-loss models with probabilities bounded away from zero and one, $d_\ell(u,v)\le C_\ell |u-v|$ for a finite constant $C_\ell$.

The calibration residual is
\begin{equation}
  \varepsilon_{\mathrm{cal}}(\rho)
  =
  \E_{f\sim \rho}
  \E_{Z\sim P_Z^r(f)}
  \sum_{a=0}^{1}
  \pi^r(a\mid Z)\,
  d_\ell\!\left\{
    \eta_\star^r(a,Z),
    \eta_{\mathrm{deb},f}^o(a,Z)
  \right\}.
  \label{eq:epsilon-cal-definition}
\end{equation}
It is nonnegative by construction. It is zero under the noiseless-calibration ideal
\begin{equation}
  \eta_\star^r(a,z)
  =
  \eta_{\mathrm{deb},f}^o(a,z)
  \quad
  \text{for }a\in\{0,1\}
\end{equation}
for $\rho$-almost every $f$ and $P_Z^r(f)$-almost every $z$.

\subsubsection{Deterministic calibration and alignment step}

\begin{lemma}[Calibration and IPM step]
\label{lem:calibration-ipm-step}
For each $f$, define the debiased OS conditional risk
\begin{equation}
  \widetilde r_f^o(z)
  =
  \sum_{a=0}^{1}
  \pi^r(a\mid z)\,
  \mathcal L_\ell\{\eta_f^r(a,z),\eta_{\mathrm{deb},f}^o(a,z)\}.
\end{equation}
Assume that $z\mapsto \widetilde r_f^o(z)/L$ belongs to the discriminator class $\calG$ for $\rho$-almost every $f$. Then
\begin{equation}
  R^r(\rho)
  \le
  \widetilde R^o(\rho)
  +L\Delta_Z(\rho)
  +\varepsilon_{\mathrm{cal}}(\rho),
  \label{eq:ipm-calibration-bound}
\end{equation}
where
\begin{align}
  \widetilde R^o(\rho)
  &=
  \E_{f\sim \rho}
  \E_{Z\sim P_Z^o(f)}
  \widetilde r_f^o(Z),\\
  \Delta_Z(\rho)
  &=
  \E_{f\sim \rho}
  \ipm_{\calG}\{P_Z^r(f),P_Z^o(f)\}.
\end{align}
\end{lemma}

\begin{proof}
For fixed $f$ and $z$, the definition of $d_\ell$ gives
\begin{align}
  &\sum_{a=0}^{1}\pi^r(a\mid z)
  \mathcal L_\ell\{\eta_f^r(a,z),\eta_\star^r(a,z)\}\\
  &\quad\le
  \widetilde r_f^o(z)
  +
  \sum_{a=0}^{1}
  \pi^r(a\mid z)
  d_\ell\!\left\{
    \eta_\star^r(a,z),
    \eta_{\mathrm{deb},f}^o(a,z)
  \right\}.
\end{align}
Integrating this inequality over $Z\sim P_Z^r(f)$ and then over $f\sim \rho$ gives the calibration part. The IPM definition and the assumption $\widetilde r_f^o/L\in\calG$ give
\begin{equation}
  \E_{P_Z^r(f)}\widetilde r_f^o(Z)
  \le
  \E_{P_Z^o(f)}\widetilde r_f^o(Z)
  +
  L\,\ipm_{\calG}\{P_Z^r(f),P_Z^o(f)\}.
\end{equation}
Another integration over $f\sim \rho$ proves Eq.~\eqref{eq:ipm-calibration-bound}. Fubini's theorem applies because the loss is bounded in Theorem~\ref{thm:pacbayes}; in the sub-Gaussian variant below it applies after truncation and then monotone convergence.
\end{proof}

\begin{remark}[Sufficient conditions for $\widetilde r_f^o/L\in\calG$]
\label{rem:sufficient-conditions-rdb}
Let $v_{a,f}(z)=\eta_{\mathrm{deb},f}^o(a,z)$, and suppose $\pi^r(a\mid z)$ is constant in $z$, as in a randomized or uniform target distribution. For squared loss with conditional variance not depending on $z$, suppose $\|\eta_f^r(a,\cdot)\|_\infty$, $\|\eta_\star^o(a,\cdot)\|_\infty$, $\|b_{0,f}\|_\infty$, and $\|b_{\Delta,f}\|_\infty$ are at most $M$, and the same four maps are $L_\eta$-Lipschitz in $z$ for $\rho$-almost every $f$. Then, for $a\in\{0,1\}$,
\begin{equation}
  \left|
    \{\eta_f^r(a,z)-v_{a,f}(z)\}^2
    -
    \{\eta_f^r(a,z')-v_{a,f}(z')\}^2
  \right|
  \le
  2(4M)(4L_\eta)\|z-z'\|
  =
  32ML_\eta\|z-z'\|.
\end{equation}
Thus $\widetilde r_f^o$ is $32ML_\eta$-Lipschitz, since $\sum_a\pi^r(a)=1$. If $\pi^r(a\mid z)$ varies with $z$, the product-rule terms involving the Lipschitz constants of $\pi^r(a\mid z)$ and the suprema of the arm-specific risks must also be included.

For Bernoulli log-loss, suppose the predicted probabilities $p_{a,f}(z)=\eta_f^r(a,z)$ and the debiased probabilities $v_{a,f}(z)$ lie in $[\epsilon,1-\epsilon]$, where $0<\epsilon<1/2$, and the component maps defining $p_{a,f}$ and $v_{a,f}$ are $L_\eta$-Lipschitz. For $h(p,v)=-v\log p-(1-v)\log(1-p)$,
\begin{equation}
  \left|\frac{\partial h}{\partial p}\right|
  =
  \left|\frac{p-v}{p(1-p)}\right|
  \le
  \frac{1}{\epsilon(1-\epsilon)},
  \qquad
  \left|\frac{\partial h}{\partial v}\right|
  =
  \left|\log\frac{1-p}{p}\right|
  \le
  \log\frac{1-\epsilon}{\epsilon}.
\end{equation}
Also $\operatorname{Lip}(p_{a,f})\le L_\eta$ and $\operatorname{Lip}(v_{a,f})\le(2+a)L_\eta\le3L_\eta$, so each arm risk is Lipschitz with constant
\begin{equation}
  \frac{L_\eta}{\epsilon(1-\epsilon)}
  +3L_\eta\log\frac{1-\epsilon}{\epsilon}.
\end{equation}
The same constant applies to $\widetilde r_f^o$ under a constant target distribution. For a bounded-Lipschitz discriminator unit ball, choose $L$ at least the relevant Lipschitz constant, and at least the sup-norm bound if the class also enforces $\|g\|_\infty\le1$.
\end{remark}

\begin{theorem}[IPM-calibration target-risk decomposition]
\label{thm:pacbayes-os}
Under the assumptions of Lemma~\ref{lem:calibration-ipm-step}, every posterior $\rho$ satisfies
\begin{equation}
  R^r(\rho)
  \le
  \widetilde R^o(\rho)
  +L\Delta_Z(\rho)
  +\varepsilon_{\mathrm{cal}}(\rho).
\end{equation}
\end{theorem}

\begin{proof}
This is Eq.~\eqref{eq:ipm-calibration-bound} from Lemma~\ref{lem:calibration-ipm-step}.
\end{proof}

\subsubsection{Proof of Theorem~\ref{thm:pacbayes}}

\begin{proof}[Proof of Theorem~\ref{thm:pacbayes}]
If $\rho$ is not absolutely continuous with respect to $\rho_0$, then $\KL(\rho\|\rho_0)=\infty$ and the claim is trivial. For absolutely continuous $\rho$, the McAllester PAC-Bayes inequality for losses in $[0,1]$ states that, with probability at least $1-\delta$ over the RCT sample, all posteriors $\rho$ satisfy
\begin{equation}
  \operatorname{kl}\{\widehat R^r(\rho),R^r(\rho)\}
  \le
  \frac{\KL(\rho\|\rho_0)+\log(2\sqrt{n^r}/\delta)}{n^r},
\end{equation}
where $\operatorname{kl}(u,v)$ is the binary relative entropy. The elementary inequality $\operatorname{kl}(u,v)\ge2(u-v)^2$ for $v\ge u$ gives
\begin{equation}
  R^r(\rho)
  \le
  \widehat R^r(\rho)
  +
  \sqrt{
    \frac{\KL(\rho\|\rho_0)+\log(2\sqrt{n^r}/\delta)}
         {2n^r}
  }.
  \label{eq:mcallester-step}
\end{equation}
This proves the first term inside the minimum in Eq.~\eqref{eq:pac_bound_rct}. The IPM-calibration decomposition is Theorem~\ref{thm:pacbayes-os}. Since both inequalities hold for the same target risk, the target risk is bounded by the better of the two displays.
\end{proof}

\begin{remark}[When the calibration residual vanishes with alignment]
\label{rem:epsilon-cal-vanishes}
Suppose $d_\ell(u,v)\le C_\ell|u-v|$, $\calG$ contains the unit ball of bounded Lipschitz functions on $\calZ$, and the function
\begin{equation}
  g_f(z)
  =
  \sum_{a=0}^{1}
  \pi^r(a\mid z)
  d_\ell\!\left\{
    \eta_\star^r(a,z),
    \eta_{\mathrm{deb},f}^o(a,z)
  \right\}
\end{equation}
satisfies $g_f/L_{\mathrm{cal}}\in\calG$ for $\rho$-almost every $f$. If the ELBO penalty $\widehat{\mathcal A}$ is a consistent estimator of $\Delta_Z(\rho)$ and the OS calibration fit gives
\begin{equation}
  \E_{f\sim \rho}\E_{Z\sim P_Z^o(f)}g_f(Z)\to0,
\end{equation}
then
\begin{equation}
  \varepsilon_{\mathrm{cal}}(\rho)
  \le
  \E_{f\sim \rho}\E_{Z\sim P_Z^o(f)}g_f(Z)
  +L_{\mathrm{cal}}\Delta_Z(\rho)
  \to0
\end{equation}
whenever $\Delta_Z(\rho)\to0$. This is the precise sense in which alignment alone is not enough: the debiased OS surface must also be calibrated on the shared support.
\end{remark}

\subsubsection{Sub-Gaussian variant}

\begin{proposition}[Sub-Gaussian PAC-Bayes variant]
\label{prop:subgaussian-pacbayes}
Replace the bounded-loss assumption in Theorem~\ref{thm:pacbayes} by the following condition: for every $f$ in the support of the data-independent prior $\rho_0$, the random variable $\ell^r(f;\xi)-R^r(f)$ is $\sigma^2$-sub-Gaussian under the RCT distribution, where $\sigma^2$ is the tail proxy. Also assume the deterministic IPM-calibration conditions of Lemma~\ref{lem:calibration-ipm-step}, with $L$, $\Delta_Z(\rho)$, $\widetilde R^o(\rho)$, and $\varepsilon_{\mathrm{cal}}(\rho)$ defined there. Then, with probability at least $1-\delta$ over the RCT sample, every posterior $\rho$ simultaneously satisfies the sub-Gaussian RCT-only PAC-Bayes bound
\begin{equation}
  R^r(\rho)\le
  \widehat R^r(\rho)
  +\sigma
  \sqrt{
    \frac{2\{\KL(\rho\|\rho_0)+\log(2\sqrt{n^r}/\delta)\}}
         {n^r}
  },
  \label{eq:subgaussian-pac-bound-rct}
\end{equation}
where $\KL(\rho\|\rho_0)$ is the Kullback-Leibler divergence from $\rho$ to $\rho_0$, and the deterministic IPM-calibration decomposition
\begin{equation}
  R^r(\rho)\le
  \widetilde R^o(\rho)
  +L\Delta_Z(\rho)
  +\varepsilon_{\mathrm{cal}}(\rho).
  \label{eq:subgaussian-pac-bound-os}
\end{equation}
Consequently the target risk is bounded by
\begin{equation}
  R^r(\rho)
  \le
  \min\left\{
    \widehat R^r(\rho)
    +\sigma
    \sqrt{
      \frac{2\{\KL(\rho\|\rho_0)+\log(2\sqrt{n^r}/\delta)\}}
           {n^r}
    },
    \widetilde R^o(\rho)+L\Delta_Z(\rho)+\varepsilon_{\mathrm{cal}}(\rho)
  \right\}.
  \label{eq:subgaussian-pac-bound}
\end{equation}
The rate remains $\mathcal{O}((n^r)^{-1/2})$; the cost is the tail constant $\sigma$ and any constants needed to verify sub-Gaussianity for the chosen loss.
\end{proposition}

\begin{proof}
For fixed $f$, independence of the RCT sample and the sub-Gaussian assumption imply
\begin{equation}
  \E_{\DR}
  \exp\left[
    \lambda\{R^r(f)-\widehat R^r(f)\}
  \right]
  \le
  \exp\left(\frac{\lambda^2\sigma^2}{2n^r}\right).
\end{equation}
Integrating with respect to the prior $\rho_0$ and applying Markov's inequality gives, with probability at least $1-\delta$,
\begin{equation}
  \E_{f\sim \rho_0}
  \exp\left[
    \lambda\{R^r(f)-\widehat R^r(f)\}
    -\frac{\lambda^2\sigma^2}{2n^r}
  \right]
  \le
  \delta^{-1}.
\end{equation}
The Donsker-Varadhan variational inequality then gives, simultaneously for all $\rho$,
\begin{equation}
  R^r(\rho)-\widehat R^r(\rho)
  \le
  \frac{\KL(\rho\|\rho_0)+\log(1/\delta)}{\lambda}
  +\frac{\lambda\sigma^2}{2n^r}.
\end{equation}
Optimizing over $\lambda>0$ yields the displayed square-root term with $\log(1/\delta)$. Replacing this logarithm by the slightly larger McAllester-style $\log(2\sqrt{n^r}/\delta)$ gives Eq.~\eqref{eq:subgaussian-pac-bound-rct}. The deterministic IPM-calibration decomposition in Eq.~\eqref{eq:subgaussian-pac-bound-os} is Theorem~\ref{thm:pacbayes-os}. Its proof uses Lemma~\ref{lem:calibration-ipm-step}, whose argument depends on the loss discrepancy $d_\ell$ and the condition $\widetilde r_f^o/L\in\calG$, not on bounded loss; for sub-Gaussian losses the required integrations follow by the truncation and monotone-convergence argument described in the lemma proof. Since both inequalities hold for the same target risk $R^r(\rho)$, the combined bound is Eq.~\eqref{eq:subgaussian-pac-bound}.
\end{proof}

\subsubsection{Connection to bias-limited information}

\begin{remark}[Connection to Theorem~\ref{thm:functional}]
\label{rem:pacbayes-functional-connection}
Consider the local Gaussian linearization in Theorem~\ref{thm:functional}. Let $\beta_\tau$ have prior covariance $\Sigma_\tau$, let the RCT contribution to the local posterior precision be $I^r$, and let the OS contribution be $I^o$. Theorem~\ref{thm:functional} gives that both $I^o$ and $\Sigma_b^{-1}-I^o$ are positive semidefinite. If $\Lambda^r=\Sigma_\tau^{-1}+I^r$ and $\Lambda^{r,o}=\Lambda^r+I^o$, then the covariance part of the Gaussian complexity satisfies
\begin{align}
  &\frac{1}{2}
  \left[
    \log\det(\Lambda^{r,o})-\log\det(\Lambda^r)
  \right]\\
  &\quad=
  \frac{1}{2}\log\det(I+(\Lambda^r)^{-1/2}I^o(\Lambda^r)^{-1/2})\\
  &\quad\le
  \frac{1}{2}\log\det(I+(\Lambda^r)^{-1/2}\Sigma_b^{-1}(\Lambda^r)^{-1/2}).
\end{align}
Thus, apart from the usual posterior-mean term, which is controlled by the same prior-norm clipping or localization assumptions used in PAC-Bayes with unbounded parameter spaces, the OS can increase the PAC-Bayes complexity by at most the displayed constant. If $\Sigma_b=\sigma_\Delta^2 I_d$, this constant is bounded by
\begin{equation}
  \frac{d}{2}\log\left\{1+\frac{\lambda_{\max}((\Lambda^r)^{-1})}{\sigma_\Delta^2}\right\}.
\end{equation}
The additional contribution to the first square-root term in Eq.~\eqref{eq:pac_bound_rct} is therefore at most this constant divided by $2n^r$ inside the square root, up to the stated mean-localization term. This is the PAC-Bayes counterpart of the bias-limited information bound: adding the OS likelihood cannot create more treatment-effect precision than the comparative-bias prior permits.
\end{remark}

\subsection{Binary-outcome notation}
\label{app:binary-outcomes}

For binary outcomes, let $g$ be a generalized-linear-model link function and let $h=g^{-1}$. The same bias decomposition can be placed on the linear-predictor scale
\begin{align}
  \eta^r(a,z) &= \mu(z)+a\tau(z),\\
  \eta^o(a,z) &= \mu(z)+a\tau(z)+b_0(z)+a b_\Delta(z),
\end{align}
with
\begin{equation}
  \Prob(Y^s=1\mid A=a,Z=z)=h\{\eta^s(a,z)\}.
\end{equation}
For a logit link, $h=\expit$. In that case $\tau(z)$ is a conditional log-odds contrast, not a risk-difference CATE. Posterior risks for trial units are
\begin{equation}
  \rho_a(x^r)=
  \E\left[h\{\mu(Z)+a\tau(Z)\}\mid Z\sim q_{\phi^r}(\cdot\mid x^r),\D\right],
\end{equation}
and the reported risk-difference CATE is
\begin{equation}
  \tau_{RD}(x^r)=\rho_1(x^r)-\rho_0(x^r).
  \label{eq:posterior_rd}
\end{equation}
Prior calibration for $b_\Delta$ should be interpreted in the units induced by the chosen link function; for example, under a logit link it is a prior on the OS-versus-RCT log-odds-ratio bias. We do not evaluate binary outcomes in this paper.

\subsection{Prior calibration for comparative bias}

For the continuous-outcome setting emphasized in the main text, a pointwise normal prior statement such as
\begin{equation}
  \Prob\{|b_\Delta(z)|<\delta\}=0.95
\end{equation}
corresponds to $\sigma_\Delta\approx\delta/1.96$. If the analyst believes the OS treatment contrast may be biased by as much as $0.4$ outcome units in either direction, then $\sigma_\Delta\approx0.20$. The analysis should report a prespecified sensitivity set, for example $\delta\in\{0,0.2,0.4,0.8,\infty\}$, and a robust mixture prior that allows posterior retreat from the optimistic component.

\subsection{Data governance and broader impacts}
\label{app:governance-impacts}

The real-data analysis uses previously collected pediatric-obesity trial and EHR control data under the source studies' data-use and institutional-review constraints; no new participants were recruited and no crowdworkers were used. Any participant-risk, consent, and disclosure procedures were handled under those source-study reviews. Patient-level data are not released. The intended benefit is more cautious external-control augmentation for small trials, especially when observational cohorts are large but imperfectly comparable. The main risk is overinterpretation: calibrated borrowing can quantify sensitivity to bias priors, but it does not remove the need for an explicitly trial-defined estimand, domain review, and study-specific validation before clinical or policy use.

\subsection{Future directions}
\label{app:future-directions}

Proposition~\ref{prop:encoder-aware-ess} gives an encoder-aware extension of Corollary~\ref{thm:scalar} in the linear-Gaussian case, adding the encoder-offset variance to the comparative-bias variance in the OS effective sample size. Several extensions follow directly. (i) \emph{Time-varying treatments}: extending the bias decomposition to longitudinal settings where the comparative bias may differ across treatment regimes raises identifiability questions tied to the structural-nested model literature. (ii) \emph{Decision-aware borrowing}: replacing the symmetric prior on $b_\Delta$ with a decision-theoretic prior that penalizes errors in the direction that would change a treatment recommendation. (iii) \emph{Multi-source fusion}: more than one observational study with study-specific bias functions and a hierarchical prior across sources.

\subsection{Reproducibility}
\label{app:reproducibility}

\paragraph{Design grids and seeds.} The factorial main simulation uses trial sizes $n^r\in\{250,500\}$, observational size $n^o=5{,}000$, and $50$ seeds per cell over the $24$-cell condition grid of Appendix~\ref{app:dgp-synthetic}. The saturation experiment uses $n^r\in\{250,500\}$, $n^o\in\{500,1{,}000,2{,}000,5{,}000,10{,}000,20{,}000,50{,}000\}$, fixed $\sigma_\Delta\in\{0.1,0.2,0.4,\infty\}$, and $50$ seeds per cell. The semi-synthetic benchmark uses $n^r=500$, $n^o=5{,}000$, and $50$ seeds.

\paragraph{Prior-scale settings across studies.} The main factorial grid selects $\sigma_\Delta$ per replicate by empirical-Bayes model averaging (Algorithm~\ref{alg:bcalm}; prior mean $m_\Delta=0$). The semi-synthetic and calibration studies use the same empirical-Bayes averaging. The saturation experiment fixes $\sigma_\Delta$ per curve and sets the comparative-bias prior mean to $m_\Delta=0.2$, which matches the location of the true bias template; its coverage values are conservative and should not be read as evidence about the adaptive procedure. The real-data application holds $\sigma_\Delta$ diffuse and selects $\sigma_0$ as described in Appendix~\ref{app:gps-details}.

\paragraph{Baseline implementations.} CF (the Wager causal forest) uses double machine learning with random-forest nuisance models. The Hahn Bayesian causal forest (BCF) is run only in the real-data analysis; in our implementation it uses two BART components (prognostic $f$ and moderating $\tau$) with explicit PGBART step methods, $50$ trees per component, and $500$ posterior draws after $250$ tuning steps. The R-OSCAR, MR-OSCAR, and CALM baselines reuse the parent calibrated-borrowing implementations with a 200-replicate trial bootstrap for intervals; in the Gaussian finite-basis comparison they coincide to within Monte Carlo noise and are reported as a single OSCAR/CALM entry.

\paragraph{Compute and code.} All results were produced on a single workstation. Code for the synthetic and semi-synthetic experiments---including simulation seeds, hyperparameter configurations, an ordered reproduction notebook, and figure-generation scripts---is available at \url{https://github.com/AsiaeeLab/b-calm}; per-source IRB and data-use constraints prevent release of patient-level data.

\section{Experimental details}
\label{app:exp-details}

This appendix gives the details needed to reproduce the experiments in Section~\ref{sec:experiments}: metric definitions (Appendix~\ref{app:metrics}), the synthetic data-generating process and condition grid (Appendix~\ref{app:dgp-synthetic}), the B-CALM estimator exactly as fit in all experiments (Appendix~\ref{app:bcalm-implementation}), the semi-synthetic construction routine (Appendix~\ref{app:semi-construction}), and the real-data cohort and configuration (Appendix~\ref{app:gps-details}). Software, seeds, and compute are collected in Appendix~\ref{app:reproducibility}.

\subsection{Metrics}
\label{app:metrics}

Each simulated replicate uses an evaluation sample $x_1,\ldots,x_m$ with $m=600$ fresh draws from the trial covariate distribution and true CATEs $\tau(x_j)$ known by construction. For a method reporting posterior-mean CATE $\widehat\tau(\cdot)$ and posterior standard deviation $\widehat s(\cdot)$, the pointwise 90\% interval is $\widehat\tau(x_j)\pm z\,\widehat s(x_j)$ with $z=z_{0.95}\approx1.645$, and we record
\begin{equation}
\mathrm{RMSE}=\Big\{\tfrac{1}{m}\textstyle\sum_{j}\{\widehat\tau(x_j)-\tau(x_j)\}^2\Big\}^{1/2},
\qquad
\mathrm{Cov}=\tfrac{1}{m}\textstyle\sum_{j}\one\{|\widehat\tau(x_j)-\tau(x_j)|\le z\,\widehat s(x_j)\},
\end{equation}
and $\mathrm{Width}=\tfrac{1}{m}\sum_j 2z\,\widehat s(x_j)$. Negative transfer in a replicate is the indicator $\one\{\mathrm{RMSE}_{\mathrm{method}}>\mathrm{RMSE}_{\text{RCT-only}}\}$ computed against the RCT-only fit of the same replicate; the reported NT is its average over replicates. For B-CALM the posterior of $\tau(x)$ is Gaussian with closed-form moments (Appendix~\ref{app:bcalm-implementation}), so $\widehat\tau$ and $\widehat s$ are analytic; for the empirical-Bayes mixture, we use the exact mixture mean and variance. CF reports asymptotic influence-function standard errors; the OSCAR/CALM family reports intervals from a 200-replicate trial bootstrap.

\subsection{Synthetic data-generating process}
\label{app:dgp-synthetic}

\paragraph{Latent state and covariates.}
For each unit, draw a latent state $z=(z_1,z_2,z_3)^\top\sim\Normal(0,I_3)$. Observational units receive a regime-specific mean shift $z\mapsto z+\delta$ with
$\delta=(0.05,-0.02,0.00)$ (shared-only), $\delta=(0.25,-0.15,0.10)$ (balanced), or $\delta=(0.60,-0.40,0.25)$ (mismatch-heavy).
Five observed covariates are generated as noisy linear maps of $z$ with independent noise $\epsilon_k\sim\Normal(0,s^2)$, where $s=0.10$, $0.25$, $0.45$ by regime. The first two covariates are common to both sources and form the shared block $W$:
\begin{equation}
x_1=z_1+\epsilon_1,\qquad x_2=z_2+\epsilon_2 .
\end{equation}
The remaining three covariates have source-specific loadings $\kappa^s$:
\begin{equation}
x_3=\kappa_3^s\,z_3+\epsilon_3,\qquad
x_4=\kappa_4^s\,z_1+\epsilon_4,\qquad
x_5=\kappa_5^s\,z_2+\epsilon_5 ,
\end{equation}
with, per regime,
\begin{center}
\small
\begin{tabular}{lcccccc}
\toprule
 & \multicolumn{2}{c}{$\kappa_3$} & \multicolumn{2}{c}{$\kappa_4$} & \multicolumn{2}{c}{$\kappa_5$}\\
Regime & RCT & OS & RCT & OS & RCT & OS\\
\midrule
shared-only\textsuperscript{*} & 1.0 & 1.0 & \multicolumn{2}{c}{$0.6z_1-0.2z_2$} & 0.4 & 0.4\\
balanced & 1.0 & 1.0 & 0.70 & 0.45 & 0.25 & 0.40\\
mismatch-heavy & 0.90 & 0.35 & 0.90 & 0.15 & 0.10 & 0.80\\
\bottomrule
\end{tabular}
\end{center}
\textsuperscript{*}In the shared-only regime both sources use the same maps ($x_4=0.6z_1-0.2z_2+\epsilon_4$), so all five covariates are effectively shared; the balanced and mismatch-heavy regimes make $x_3,x_4,x_5$ progressively more source-specific, which creates the trial-only block $U$ and OS-only block $V$.

\paragraph{Outcome surfaces.}
Outcome complexity is tied to the regime: shared-only uses linear surfaces, balanced adds mild nonlinearity, and mismatch-heavy adds further nonlinearity. With covariates $x=(x_1,\ldots,x_5)$,
\begin{align}
\mu(x)&=0.25+0.45x_1-0.25x_2+0.18x_3+0.08x_4
\;{+}\;\underbrace{0.12\sin x_1-0.08x_2x_3}_{\text{balanced, mismatch-heavy}}
\;{+}\;\underbrace{0.10\cos x_4+0.08\tanh(x_1x_5)}_{\text{mismatch-heavy}},\\
\tau(x)&=0.35+0.16x_1-0.10x_3
\;{+}\;\underbrace{0.08\sin x_2+0.05x_1x_4}_{\text{balanced, mismatch-heavy}}
\;{+}\;\underbrace{0.10\tanh x_5-0.06\cos(x_1+x_3)}_{\text{mismatch-heavy}}.
\end{align}

\paragraph{Bias functions and treatment assignment.}
Under the no-bias regime, $b_0=b_\Delta\equiv0$. Otherwise the OS has baseline bias
$b_0(x)=0.20+0.15x_1-0.10x_3$,
and under the comparative regime additionally
\begin{equation}
b_\Delta(x)=\frac{s_\Delta}{0.2}\left\{0.20+0.12x_1-0.08x_2+0.04\tanh x_4\right\},
\label{eq:dgp-bdelta}
\end{equation}
where $s_\Delta$ is the comparative-bias strength ($s_\Delta=0.2$ except in the sweep, which uses $s_\Delta\in\{0,0.2,0.5,1.0\}$); the baseline-only regime sets $b_\Delta\equiv0$. Trial treatment is assigned by balanced randomization ($\lfloor n^r/2\rfloor$ treated, randomly permuted). OS treatment is confounded:
\begin{equation}
A^o\sim\Bern\left[\expit\{-0.05+0.85x_1-0.55x_2+0.35x_4+0.35x_3\,\one(\text{comparative})\}\right].
\label{eq:dgp-assignment}
\end{equation}
Outcomes are $Y=\mu(X)+A\,\tau(X)+\one(\text{OS})\{b_0(X)+A\,b_\Delta(X)\}+\varepsilon$ with $\varepsilon\sim\Normal(0,1)$.

\paragraph{Condition grid and evaluation.}
The factorial grid crosses the three regimes with the three bias settings at $s_\Delta=0.2$, plus the balanced-regime comparative sweep $s_\Delta\in\{0,0.5,1.0\}$, each at $n^r\in\{250,500\}$ with $n^o=5{,}000$: $24$ cells, $50$ seeds per cell. Every replicate is evaluated on $m=600$ fresh trial-distribution covariate draws with known $\tau(\cdot)$. OS-only is omitted from the three sweep-extension cells.

\begin{algorithm}[t]
\caption{Synthetic data generation (one replicate)}
\label{alg:dgp}
\begin{algorithmic}[1]
\Require regime $r\in\{$shared-only, balanced, mismatch-heavy$\}$, bias setting, strength $s_\Delta$, sizes $n^r,n^o$, seed
\State draw latent states $z\sim\Normal(0,I_3)$ for $n^r$ trial and $n^o$ OS units; shift OS states by $\delta_r$
\State map to covariates $x_1,\ldots,x_5$ with regime- and source-specific loadings and noise scale $s_r$
\State assign trial treatment by balanced randomization; draw OS treatment from Eq.~\eqref{eq:dgp-assignment}
\State compute $\mu(x),\tau(x)$ at regime complexity; compute $b_0(x),b_\Delta(x)$ per bias setting and Eq.~\eqref{eq:dgp-bdelta}
\State draw $Y=\mu+A\tau+\one(\text{OS})(b_0+Ab_\Delta)+\Normal(0,1)$
\State draw $m=600$ evaluation covariates from the trial law; store true $\tau(\cdot)$ there
\end{algorithmic}
\end{algorithm}

\subsection{B-CALM implementation and pseudocode}
\label{app:bcalm-implementation}

\paragraph{Model.}
All experiments use the closed-form Gaussian instance of Section~\ref{sec:model} with identity encoders: each source's covariate vector is its latent state, and $b_0$ and $b_\Delta$ absorb the remaining source discrepancy. Fix a basis map $\phi(\cdot)\in\R^{k}$ and write $\mu(x)=\phi(x)^\top\omega_\mu$, $\tau(x)=\phi(x)^\top\omega_\tau$, $b_0(x)=\phi(x)^\top\omega_0$, $b_\Delta(x)=\phi(x)^\top\omega_\Delta$. Stacking $\theta=(\omega_\mu,\omega_\tau,\omega_0,\omega_\Delta)\in\R^{4k}$, each trial row contributes the design vector $(\phi^\top, a\phi^\top, 0, 0)$ and each OS row $(\phi^\top, a\phi^\top, \phi^\top, a\phi^\top)$, with Gaussian likelihood $Y\mid\cdot\sim\Normal(\text{design}\cdot\theta,\sigma_y^2)$, $\sigma_y=1$. The basis matches the regime: linear $\phi(x)=(1,x)$ for shared-only; adding $\sin x_1$, $x_2^2$, $x_1x_3$ for balanced; adding $\cos x_4$, $\tanh(x_1x_5)$ for mismatch-heavy. The basis overlaps but does not span the true surfaces (for example, the balanced-regime truth contains $x_2x_3$, $\sin x_2$, and $x_1x_4$, none of which is in the basis), so the outcome model is mildly misspecified by design.

\paragraph{Priors and posterior.}
Coefficients have independent Gaussian priors: $\omega_\mu,\omega_\tau\sim\Normal(0,3^2I_k)$, $\omega_0\sim\Normal(0,\sigma_0^2I_k)$ with $\sigma_0=0.6$, and $\omega_\Delta\sim\Normal(m_\Delta e_1,\sigma_\Delta^2 I_k)$, where $e_1$ selects the intercept coefficient; $m_\Delta=0$ in the factorial grid and semi-synthetic study and $m_\Delta=0.2$ in the saturation experiment. Diffuse and degenerate scales are implemented numerically (prior variance $10^{10}$ for $\infty$, $10^{-10}$ for $0$). With design matrix $\Phi$ and prior precision $P$, the posterior is Gaussian with precision $\Lambda=P+\Phi^\top\Phi/\sigma_y^2$ and mean $\Lambda^{-1}(Pm_{\mathrm{prior}}+\Phi^\top Y/\sigma_y^2)$, computed by Cholesky factorization. The posterior CATE at $x$ is Gaussian with mean $\phi(x)^\top\widehat\omega_\tau$ and variance $\phi(x)^\top\mathrm{Cov}(\omega_\tau)\phi(x)$; no sampling is required.

\paragraph{Empirical-Bayes selection of $\sigma_\Delta$.}
In the factorial grid, $\sigma_\Delta$ is selected per replicate by model averaging over the fixed candidate set $S=\{0.01,0.05,0.2,0.5,1.0,\infty\}$ with fixed prior weights $\pi=(0.01,0.01,0.42,0.21,0.13,0.22)$. The weights are asymmetric by design: an overly tight $\sigma_\Delta$ over-borrows a biased contrast and threatens validity, while an overly diffuse $\sigma_\Delta$ only costs efficiency, so the prior stays skeptical on the tight scales, keeps a mild mode at moderate borrowing, and carries a heavy diffuse tail for retreat --- the robust-mixture structure recommended in the prior-calibration discussion of Appendix~\ref{app:functional-pac}. The weights are one global choice, used unchanged in every cell and study; per-replicate posterior weights are exported with the results for inspection. The evidence for each candidate is the exact Gaussian marginal likelihood of the contrast discrepancy: fit the RCT-only model (diffuse bias priors) and an OS-only contrast model, extract the two posterior contrast-coefficient distributions $(\widehat\tau^r,\Sigma^r)$ and $(\widehat\tau^o,\Sigma^o)$, and score
$\log p_\sigma = \log\Normal\{\widehat\tau^o-\widehat\tau^r;\,0,\;\Sigma^r+\Sigma^o+\sigma^2I\}$.
Posterior model weights are $w_\sigma\propto\pi_\sigma\exp(\log p_\sigma)$, and the reported posterior is the corresponding mixture of Gaussian fits with exact mixture moments. The semi-synthetic and calibration studies use the same averaging; the saturation experiment instead fixes $\sigma_\Delta\in\{0.1,0.2,0.4,\infty\}$ per curve, and the semi-synthetic sensitivity figure varies fixed $\sigma_\Delta$ over $\{0,0.1,0.2,0.4,\infty\}$.

\begin{algorithm}[t]
\caption{B-CALM, Gaussian finite-basis instance with empirical-Bayes averaging}
\label{alg:bcalm}
\begin{algorithmic}[1]
\Require trial data $\DR$, OS data $\DO$, basis $\phi$, candidate set $S$, prior weights $\pi$, $\sigma_0=0.6$, $\sigma_y=1$
\State build design rows $(\phi,a\phi,0,0)$ for trial units and $(\phi,a\phi,\phi,a\phi)$ for OS units
\State fit RCT-only ($\sigma_\Delta=\sigma_0=\infty$, trial rows only) and OS-only contrast ($\sigma_\Delta=0$, OS rows only); extract contrast posteriors $(\widehat\tau^r,\Sigma^r)$, $(\widehat\tau^o,\Sigma^o)$
\For{$\sigma\in S$}
  \State fit the joint Gaussian posterior with comparative-bias scale $\sigma$ (closed form, Cholesky)
  \State compute $\log p_\sigma=\log\Normal\{\widehat\tau^o-\widehat\tau^r;0,\Sigma^r+\Sigma^o+\sigma^2I\}$
\EndFor
\State set weights $w_\sigma\propto\pi_\sigma\exp(\log p_\sigma)$; return the weighted mixture of fits
\State report posterior CATE mean/variance at any $x$ from exact mixture moments; intervals are mean $\pm\,1.645\,$sd
\end{algorithmic}
\end{algorithm}

\paragraph{Baseline configurations of the same model.}
RCT-only fits the trial rows alone with diffuse bias priors; OS-only fits the OS rows alone with $\sigma_0=\sigma_\Delta=0$ (its contrast is used directly); pooled fits both sources with $\sigma_0=\sigma_\Delta=0$, forcing $b_0=b_\Delta\equiv0$.

\subsection{Semi-synthetic construction modeled on STAR}
\label{app:semi-construction}

The benchmark keeps STAR-type covariates and marginals but imposes known outcome surfaces, following \citet{asiaee2025roscar}. For source $s\in\{r,o\}$, covariates are drawn as
sex $\sim\Bern(0.49)$; race (Black) $\sim\Bern(0.32)$ for the trial and $\Bern(0.38)$ for the OS; free lunch $\sim\Bern(0.42)$ vs.\ $\Bern(0.50)$; small class $\sim\Bern(0.33)$ vs.\ $\Bern(0.24)$; teacher experience $\sim\Normal(8,3^2)$ vs.\ $\Normal(7,3^2)$; baseline score $=\Normal(0,1)-0.35\,\text{freelunch}+0.018\,\text{teacherexp}$; and school quality $=\Normal(0,0.7^2)+0.2(1-\text{freelunch})$. The seven covariates enter centered and scaled as
$x=(\text{baseline score},\ \text{sex}-0.5,\ \text{race}-0.32,\ \text{freelunch}-0.45,\ \text{smallclass}-0.30,\ (\text{teacherexp}-8)/3,\ \text{school quality})$.
The imposed surfaces are
\begin{align}
\mu(x)&=0.15+0.60x_1+0.12x_2-0.22x_4+0.16x_6+0.20x_7+0.08\sin x_1,\\
\tau(x)&=0.22+0.18(1-x_4)+0.12x_5+0.08\tanh x_1,
\end{align}
with OS bias functions $b_0(x)=0.22+0.18x_4-0.12x_7$ and $b_\Delta(x)=0.18+0.10x_4-0.08x_5+0.05\tanh x_1$ (comparative strength $0.2$). Trial treatment is balanced-randomized; OS treatment is confounded via $A^o\sim\Bern[\expit\{-0.15+0.75x_1-0.55x_4+0.30x_7\}]$; outcomes add $\Normal(0,1)$ noise. Sizes are $n^r=500$, $n^o=5{,}000$, with $50$ seeds and $600$ trial-law evaluation draws. B-CALM uses the star basis ($\phi=(1,x,\sin x_1,x_2^2,x_1x_3,\cos x_4,\tanh(x_1x_5))$) with the same empirical-Bayes averaging over $\sigma_\Delta$ as the main grid; the sensitivity analysis in Figure~\ref{fig:semi-sensitivity}, which varies fixed $\sigma_\Delta$ from $0$ to $\infty$, shows how conclusions move across the whole borrowing range. Among the calibrated-borrowing family, MR-OSCAR is run in this study and represents the family. The sensitivity figure fixes one seed and traces posterior CATEs at seven representative covariate profiles (the first seven evaluation units).

\subsection{Real-data cohort and configuration}
\label{app:gps-details}

The trial is a multi-site pediatric-obesity prevention RCT with $n^r=385$ children and randomized treatment; the external cohort is a cross-sectional EHR sample of $n^o=15{,}150$ children drawn from the same health systems, none of whom received the trial intervention, so every EHR row enters with $A=0$ (an external-control convention). The outcome is the 12-month weight-for-length z-score. Six shared covariates are available in both sources: site, sex, race/ethnicity, language, insurance, and baseline weight-for-length z-score; five trial-only covariates (household food security, parent education, income, health literacy, and birth weight) are retained for sensitivity analyses only. B-CALM is fit on the six shared covariates with a linear basis. With one-arm external controls the comparative-bias scale $\sigma_\Delta$ is unidentified and held diffuse; borrowing is governed by the baseline-commensurability scale $\sigma_0$, selected over the candidate set $\{0,0.025,0.05,0.1,0.2,0.5,\infty\}$ with fixed prior weights $(0.05,0.20,0.30,0.20,0.10,0.05,0.10)$ by the prior-weighted score
$\log\pi(\sigma_0)-\log\{\text{interval width}(\sigma_0)\}-0.02\,|\widehat{\ATE}(\sigma_0)|$,
which is maximized at $\sigma_0=0.05$. The skeptical scale $\sigma_0=0.6$ used in the bias-susceptibility diagnostic is the default baseline-bias prior of the synthetic experiments, an order of magnitude above the selected value. BCF is excluded from the bias-susceptibility diagnostic because that diagnostic requires a grid of refits (roughly $2(K{+}1)$ for $K$ strata), which is infeasible at BCF's MCMC cost; it appears in the main real-data comparison, Table~\ref{tab:gps-main}. Patient-level data cannot be released (Appendix~\ref{app:governance-impacts}).

\section{Full main simulation results table}
\label{app:main-table}

The 138-row main-results table summarized by Table~\ref{tab:main_summary} in the body is reproduced in full below. Each row reports the average across $50$ seeds for one method in one combination of (DGP, bias regime, comparative-bias strength, $n^r$), using the condition grid of Appendix~\ref{app:dgp-synthetic}; RMSE is reported as mean $\pm$ across-seed standard deviation. Bold and underline preserve the generated per-cell highlight markers for quick scanning.

{\tiny
\setlength\tabcolsep{3pt}
\begin{longtable}{lllrlrrrr}
\caption{Full main simulation results over 50 seeds per cell. Within each cell we bold the method with the lowest negative-transfer rate and underline the method with coverage closest to the 90\% target from below. Pooled and CF (the Wager causal forest) achieve lower RMSE under no bias because pooling is variance-optimal there, but their coverage collapses under comparative bias. R-OSCAR, MR-OSCAR, and CALM coincide to within Monte Carlo noise in this harness and are merged into a single OSCAR/CALM row; the family sits just below nominal coverage on average while declining to borrow under bias. The Hahn Bayesian causal forest (BCF) is run only in the real-data analysis. Bias-limited width saturation is guaranteed by Corollary~\ref{thm:scalar}.}
\label{tab:main} \\
\toprule
DGP & Bias & $b_\Delta$ & $n^r$ & Method & Coverage & Neg. transfer & RMSE & Width \\
\midrule
\endfirsthead
\multicolumn{9}{l}{\emph{(continued)}} \\
\toprule
DGP & Bias & $b_\Delta$ & $n^r$ & Method & Coverage & Neg. transfer & RMSE & Width \\
\midrule
\endhead
\midrule
\multicolumn{9}{r}{\emph{continued on next page}} \\
\endfoot
\endlastfoot
balanced & baseline-only &  & 250 & pooled & 0.794 & 0.000 & 0.135 $\pm$ 0.031 & 0.285 \\
balanced & baseline-only &  & 250 & \underline{OS-only} & \underline{0.794} & 0.020 & 0.137 $\pm$ 0.034 & 0.297 \\
balanced & baseline-only &  & 250 & CF & 0.967 & 0.000 & 0.152 $\pm$ 0.016 & 0.760 \\
balanced & baseline-only &  & 250 & \textbf{B-CALM} & \textbf{0.978} & \textbf{0.000} & \textbf{0.201 $\pm$ 0.076} & \textbf{0.795} \\
balanced & baseline-only &  & 250 & OSCAR/CALM & 0.918 & 0.020 & 0.258 $\pm$ 0.075 & 0.895 \\
balanced & baseline-only &  & 250 & RCT-only & 0.926 & 0.000 & 0.370 $\pm$ 0.101 & 1.225 \\
balanced & baseline-only &  & 500 & pooled & 0.802 & 0.020 & 0.124 $\pm$ 0.024 & 0.274 \\
balanced & baseline-only &  & 500 & OS-only & 0.807 & 0.020 & 0.129 $\pm$ 0.030 & 0.297 \\
balanced & baseline-only &  & 500 & B-CALM & 0.954 & 0.000 & 0.164 $\pm$ 0.051 & 0.591 \\
balanced & baseline-only &  & 500 & OSCAR/CALM & 0.867 & 0.000 & 0.201 $\pm$ 0.056 & 0.608 \\
balanced & baseline-only &  & 500 & \textbf{CF} & \textbf{0.962} & \textbf{0.000} & \textbf{0.153 $\pm$ 0.016} & \textbf{0.741} \\
balanced & baseline-only &  & 500 & \underline{RCT-only} & \underline{0.877} & 0.000 & 0.276 $\pm$ 0.056 & 0.841 \\
balanced & comparative & 0.0 & 250 & \underline{pooled} & \underline{0.765} & 0.000 & 0.138 $\pm$ 0.031 & 0.286 \\
balanced & comparative & 0.0 & 250 & CF & 0.965 & 0.000 & 0.155 $\pm$ 0.016 & 0.764 \\
balanced & comparative & 0.0 & 250 & \textbf{B-CALM} & \textbf{0.977} & \textbf{0.000} & \textbf{0.203 $\pm$ 0.076} & \textbf{0.795} \\
balanced & comparative & 0.0 & 250 & OSCAR/CALM & 0.918 & 0.020 & 0.258 $\pm$ 0.075 & 0.895 \\
balanced & comparative & 0.0 & 250 & RCT-only & 0.926 & 0.000 & 0.370 $\pm$ 0.101 & 1.225 \\
balanced & comparative & 0.0 & 500 & pooled & 0.765 & 0.020 & 0.131 $\pm$ 0.030 & 0.275 \\
balanced & comparative & 0.0 & 500 & B-CALM & 0.952 & 0.020 & 0.167 $\pm$ 0.049 & 0.595 \\
balanced & comparative & 0.0 & 500 & OSCAR/CALM & 0.867 & 0.000 & 0.201 $\pm$ 0.056 & 0.608 \\
balanced & comparative & 0.0 & 500 & CF & 0.963 & 0.040 & 0.153 $\pm$ 0.018 & 0.748 \\
balanced & comparative & 0.0 & 500 & \textbf{\underline{RCT-only}} & \textbf{\underline{0.877}} & \textbf{0.000} & \textbf{0.276 $\pm$ 0.056} & \textbf{0.841} \\
balanced & comparative & 0.2 & 250 & B-CALM & 0.924 & 0.020 & 0.241 $\pm$ 0.060 & 0.857 \\
balanced & comparative & 0.2 & 250 & pooled & 0.308 & 0.180 & 0.298 $\pm$ 0.032 & 0.287 \\
balanced & comparative & 0.2 & 250 & OS-only & 0.320 & 0.220 & 0.295 $\pm$ 0.032 & 0.299 \\
balanced & comparative & 0.2 & 250 & \underline{CF} & \underline{0.791} & 0.200 & 0.284 $\pm$ 0.033 & 0.771 \\
balanced & comparative & 0.2 & 250 & OSCAR/CALM & 0.903 & 0.000 & 0.268 $\pm$ 0.092 & 0.903 \\
balanced & comparative & 0.2 & 250 & \textbf{RCT-only} & \textbf{0.911} & \textbf{0.000} & \textbf{0.382 $\pm$ 0.083} & \textbf{1.233} \\
balanced & comparative & 0.2 & 500 & OSCAR/CALM & 0.903 & 0.020 & 0.183 $\pm$ 0.047 & 0.610 \\
balanced & comparative & 0.2 & 500 & B-CALM & 0.928 & 0.040 & 0.192 $\pm$ 0.043 & 0.664 \\
balanced & comparative & 0.2 & 500 & pooled & 0.309 & 0.600 & 0.289 $\pm$ 0.032 & 0.276 \\
balanced & comparative & 0.2 & 500 & OS-only & 0.330 & 0.600 & 0.291 $\pm$ 0.036 & 0.300 \\
balanced & comparative & 0.2 & 500 & \textbf{RCT-only} & \textbf{0.909} & \textbf{0.000} & \textbf{0.269 $\pm$ 0.061} & \textbf{0.847} \\
balanced & comparative & 0.2 & 500 & \underline{CF} & \underline{0.796} & 0.580 & 0.277 $\pm$ 0.031 & 0.760 \\
balanced & comparative & 0.5 & 250 & OSCAR/CALM & 0.904 & 0.000 & 0.270 $\pm$ 0.068 & 0.910 \\
balanced & comparative & 0.5 & 250 & \underline{B-CALM} & \underline{0.859} & 0.140 & 0.295 $\pm$ 0.083 & 0.933 \\
balanced & comparative & 0.5 & 250 & \textbf{RCT-only} & \textbf{0.912} & \textbf{0.000} & \textbf{0.386 $\pm$ 0.111} & \textbf{1.245} \\
balanced & comparative & 0.5 & 250 & pooled & 0.136 & 0.960 & 0.663 $\pm$ 0.033 & 0.288 \\
balanced & comparative & 0.5 & 250 & CF & 0.370 & 0.960 & 0.619 $\pm$ 0.036 & 0.798 \\
balanced & comparative & 0.5 & 500 & \textbf{\underline{OSCAR/CALM}} & \textbf{\underline{0.880}} & \textbf{0.000} & \textbf{0.199 $\pm$ 0.056} & \textbf{0.609} \\
balanced & comparative & 0.5 & 500 & B-CALM & 0.857 & 0.060 & 0.239 $\pm$ 0.055 & 0.697 \\
balanced & comparative & 0.5 & 500 & RCT-only & 0.880 & 0.000 & 0.293 $\pm$ 0.065 & 0.845 \\
balanced & comparative & 0.5 & 500 & CF & 0.393 & 1.000 & 0.590 $\pm$ 0.040 & 0.786 \\
balanced & comparative & 0.5 & 500 & pooled & 0.132 & 1.000 & 0.636 $\pm$ 0.035 & 0.276 \\
balanced & comparative & 1.0 & 250 & OSCAR/CALM & 0.870 & 0.000 & 0.302 $\pm$ 0.087 & 0.917 \\
balanced & comparative & 1.0 & 250 & \underline{B-CALM} & \underline{0.893} & 0.060 & 0.346 $\pm$ 0.089 & 1.094 \\
balanced & comparative & 1.0 & 250 & \textbf{RCT-only} & \textbf{0.875} & \textbf{0.000} & \textbf{0.418 $\pm$ 0.095} & \textbf{1.245} \\
balanced & comparative & 1.0 & 250 & pooled & 0.074 & 1.000 & 1.283 $\pm$ 0.056 & 0.287 \\
balanced & comparative & 1.0 & 250 & CF & 0.224 & 1.000 & 1.185 $\pm$ 0.057 & 0.874 \\
balanced & comparative & 1.0 & 500 & OSCAR/CALM & 0.870 & 0.020 & 0.200 $\pm$ 0.059 & 0.611 \\
balanced & comparative & 1.0 & 500 & \underline{B-CALM} & \underline{0.897} & 0.080 & 0.251 $\pm$ 0.072 & 0.782 \\
balanced & comparative & 1.0 & 500 & \textbf{RCT-only} & \textbf{0.889} & \textbf{0.000} & \textbf{0.280 $\pm$ 0.075} & \textbf{0.845} \\
balanced & comparative & 1.0 & 500 & CF & 0.229 & 1.000 & 1.142 $\pm$ 0.044 & 0.853 \\
balanced & comparative & 1.0 & 500 & pooled & 0.071 & 1.000 & 1.226 $\pm$ 0.044 & 0.275 \\
balanced & none &  & 250 & \underline{pooled} & \underline{0.799} & 0.000 & 0.131 $\pm$ 0.029 & 0.285 \\
balanced & none &  & 250 & OS-only & 0.794 & 0.020 & 0.137 $\pm$ 0.034 & 0.297 \\
balanced & none &  & 250 & CF & 0.964 & 0.000 & 0.154 $\pm$ 0.017 & 0.750 \\
balanced & none &  & 250 & \textbf{B-CALM} & \textbf{0.978} & \textbf{0.000} & \textbf{0.201 $\pm$ 0.076} & \textbf{0.795} \\
balanced & none &  & 250 & OSCAR/CALM & 0.918 & 0.020 & 0.258 $\pm$ 0.075 & 0.895 \\
balanced & none &  & 250 & RCT-only & 0.926 & 0.000 & 0.370 $\pm$ 0.101 & 1.225 \\
balanced & none &  & 500 & pooled & 0.815 & 0.000 & 0.118 $\pm$ 0.022 & 0.274 \\
balanced & none &  & 500 & OS-only & 0.807 & 0.020 & 0.129 $\pm$ 0.030 & 0.297 \\
balanced & none &  & 500 & B-CALM & 0.954 & 0.000 & 0.164 $\pm$ 0.051 & 0.591 \\
balanced & none &  & 500 & OSCAR/CALM & 0.867 & 0.000 & 0.201 $\pm$ 0.056 & 0.608 \\
balanced & none &  & 500 & \textbf{CF} & \textbf{0.965} & \textbf{0.000} & \textbf{0.151 $\pm$ 0.014} & \textbf{0.738} \\
balanced & none &  & 500 & \underline{RCT-only} & \underline{0.877} & 0.000 & 0.276 $\pm$ 0.056 & 0.841 \\
\addlinespace
\multicolumn{9}{l}{Avg coverage: B-CALM 0.929; RCT-only 0.899; OSCAR/CALM 0.890; pooled 0.481; OS-only 0.642; CF 0.716} \\
\multicolumn{9}{l}{Avg neg-transfer: B-CALM 0.035; RCT-only 0.000; OSCAR/CALM 0.008; pooled 0.398; OS-only 0.150; CF 0.398} \\
\addlinespace
mismatch-heavy & baseline-only &  & 250 & pooled & 0.787 & 0.000 & 0.196 $\pm$ 0.052 & 0.434 \\
mismatch-heavy & baseline-only &  & 250 & \underline{OS-only} & \underline{0.805} & 0.000 & 0.214 $\pm$ 0.059 & 0.507 \\
mismatch-heavy & baseline-only &  & 250 & CF & 0.952 & 0.000 & 0.165 $\pm$ 0.020 & 0.747 \\
mismatch-heavy & baseline-only &  & 250 & \textbf{B-CALM} & \textbf{0.966} & \textbf{0.000} & \textbf{0.244 $\pm$ 0.081} & \textbf{0.915} \\
mismatch-heavy & baseline-only &  & 250 & OSCAR/CALM & 0.905 & 0.000 & 0.280 $\pm$ 0.074 & 0.922 \\
mismatch-heavy & baseline-only &  & 250 & RCT-only & 0.907 & 0.000 & 0.440 $\pm$ 0.110 & 1.374 \\
mismatch-heavy & baseline-only &  & 500 & pooled & 0.810 & 0.020 & 0.165 $\pm$ 0.035 & 0.392 \\
mismatch-heavy & baseline-only &  & 500 & OS-only & 0.816 & 0.160 & 0.217 $\pm$ 0.075 & 0.505 \\
mismatch-heavy & baseline-only &  & 500 & B-CALM & 0.942 & 0.000 & 0.192 $\pm$ 0.053 & 0.667 \\
mismatch-heavy & baseline-only &  & 500 & OSCAR/CALM & 0.846 & 0.000 & 0.221 $\pm$ 0.051 & 0.628 \\
mismatch-heavy & baseline-only &  & 500 & \textbf{CF} & \textbf{0.951} & \textbf{0.000} & \textbf{0.164 $\pm$ 0.018} & \textbf{0.740} \\
mismatch-heavy & baseline-only &  & 500 & \underline{RCT-only} & \underline{0.883} & 0.000 & 0.307 $\pm$ 0.065 & 0.935 \\
mismatch-heavy & comparative & 0.2 & 250 & CF & 0.818 & 0.060 & 0.262 $\pm$ 0.033 & 0.760 \\
mismatch-heavy & comparative & 0.2 & 250 & OSCAR/CALM & 0.890 & 0.040 & 0.288 $\pm$ 0.094 & 0.928 \\
mismatch-heavy & comparative & 0.2 & 250 & B-CALM & 0.926 & 0.020 & 0.275 $\pm$ 0.069 & 0.976 \\
mismatch-heavy & comparative & 0.2 & 250 & pooled & 0.404 & 0.180 & 0.338 $\pm$ 0.045 & 0.435 \\
mismatch-heavy & comparative & 0.2 & 250 & OS-only & 0.464 & 0.220 & 0.346 $\pm$ 0.055 & 0.507 \\
mismatch-heavy & comparative & 0.2 & 250 & \textbf{\underline{RCT-only}} & \textbf{\underline{0.896}} & \textbf{0.000} & \textbf{0.443 $\pm$ 0.106} & \textbf{1.385} \\
mismatch-heavy & comparative & 0.2 & 500 & OSCAR/CALM & 0.880 & 0.000 & 0.205 $\pm$ 0.042 & 0.631 \\
mismatch-heavy & comparative & 0.2 & 500 & B-CALM & 0.932 & 0.020 & 0.216 $\pm$ 0.040 & 0.759 \\
mismatch-heavy & comparative & 0.2 & 500 & CF & 0.818 & 0.300 & 0.258 $\pm$ 0.026 & 0.745 \\
mismatch-heavy & comparative & 0.2 & 500 & pooled & 0.383 & 0.600 & 0.318 $\pm$ 0.039 & 0.395 \\
mismatch-heavy & comparative & 0.2 & 500 & \textbf{\underline{RCT-only}} & \textbf{\underline{0.900}} & \textbf{0.000} & \textbf{0.305 $\pm$ 0.063} & \textbf{0.940} \\
mismatch-heavy & comparative & 0.2 & 500 & OS-only & 0.466 & 0.660 & 0.342 $\pm$ 0.053 & 0.510 \\
mismatch-heavy & none &  & 250 & pooled & 0.797 & 0.000 & 0.191 $\pm$ 0.057 & 0.434 \\
mismatch-heavy & none &  & 250 & \underline{OS-only} & \underline{0.805} & 0.000 & 0.214 $\pm$ 0.059 & 0.507 \\
mismatch-heavy & none &  & 250 & CF & 0.943 & 0.000 & 0.172 $\pm$ 0.021 & 0.739 \\
mismatch-heavy & none &  & 250 & \textbf{B-CALM} & \textbf{0.966} & \textbf{0.000} & \textbf{0.244 $\pm$ 0.081} & \textbf{0.915} \\
mismatch-heavy & none &  & 250 & OSCAR/CALM & 0.906 & 0.000 & 0.280 $\pm$ 0.074 & 0.922 \\
mismatch-heavy & none &  & 250 & RCT-only & 0.907 & 0.000 & 0.440 $\pm$ 0.110 & 1.374 \\
mismatch-heavy & none &  & 500 & pooled & 0.832 & 0.020 & 0.156 $\pm$ 0.031 & 0.392 \\
mismatch-heavy & none &  & 500 & OS-only & 0.816 & 0.160 & 0.217 $\pm$ 0.076 & 0.505 \\
mismatch-heavy & none &  & 500 & B-CALM & 0.942 & 0.000 & 0.192 $\pm$ 0.053 & 0.667 \\
mismatch-heavy & none &  & 500 & OSCAR/CALM & 0.846 & 0.000 & 0.221 $\pm$ 0.051 & 0.628 \\
mismatch-heavy & none &  & 500 & \textbf{CF} & \textbf{0.944} & \textbf{0.000} & \textbf{0.169 $\pm$ 0.022} & \textbf{0.731} \\
mismatch-heavy & none &  & 500 & \underline{RCT-only} & \underline{0.883} & 0.000 & 0.307 $\pm$ 0.065 & 0.935 \\
\addlinespace
\multicolumn{9}{l}{Avg coverage: B-CALM 0.946; RCT-only 0.896; OSCAR/CALM 0.879; pooled 0.669; OS-only 0.695; CF 0.904} \\
\multicolumn{9}{l}{Avg neg-transfer: B-CALM 0.007; RCT-only 0.000; OSCAR/CALM 0.007; pooled 0.137; OS-only 0.200; CF 0.060} \\
\addlinespace
shared-only & baseline-only &  & 250 & pooled & 0.883 & 0.000 & 0.069 $\pm$ 0.022 & 0.224 \\
shared-only & baseline-only &  & 250 & \underline{OS-only} & \underline{0.892} & 0.000 & 0.071 $\pm$ 0.021 & 0.231 \\
shared-only & baseline-only &  & 250 & \textbf{B-CALM} & \textbf{0.980} & \textbf{0.000} & \textbf{0.156 $\pm$ 0.056} & \textbf{0.683} \\
shared-only & baseline-only &  & 250 & CF & 0.966 & 0.080 & 0.156 $\pm$ 0.016 & 0.759 \\
shared-only & baseline-only &  & 250 & OSCAR/CALM & 0.918 & 0.040 & 0.226 $\pm$ 0.076 & 0.804 \\
shared-only & baseline-only &  & 250 & RCT-only & 0.936 & 0.000 & 0.265 $\pm$ 0.078 & 0.977 \\
shared-only & baseline-only &  & 500 & pooled & 0.904 & 0.000 & 0.064 $\pm$ 0.023 & 0.218 \\
shared-only & baseline-only &  & 500 & OS-only & 0.915 & 0.000 & 0.067 $\pm$ 0.021 & 0.231 \\
shared-only & baseline-only &  & 500 & \textbf{B-CALM} & \textbf{0.944} & \textbf{0.000} & \textbf{0.137 $\pm$ 0.055} & \textbf{0.516} \\
shared-only & baseline-only &  & 500 & OSCAR/CALM & 0.869 & 0.000 & 0.177 $\pm$ 0.059 & 0.548 \\
shared-only & baseline-only &  & 500 & CF & 0.966 & 0.240 & 0.152 $\pm$ 0.017 & 0.748 \\
shared-only & baseline-only &  & 500 & \underline{RCT-only} & \underline{0.887} & 0.000 & 0.212 $\pm$ 0.062 & 0.692 \\
shared-only & comparative & 0.2 & 250 & B-CALM & 0.926 & 0.040 & 0.187 $\pm$ 0.070 & 0.714 \\
shared-only & comparative & 0.2 & 250 & pooled & 0.273 & 0.500 & 0.266 $\pm$ 0.033 & 0.225 \\
shared-only & comparative & 0.2 & 250 & OS-only & 0.270 & 0.500 & 0.265 $\pm$ 0.033 & 0.232 \\
shared-only & comparative & 0.2 & 250 & OSCAR/CALM & 0.910 & 0.060 & 0.234 $\pm$ 0.085 & 0.813 \\
shared-only & comparative & 0.2 & 250 & \underline{CF} & \underline{0.794} & 0.560 & 0.280 $\pm$ 0.034 & 0.774 \\
shared-only & comparative & 0.2 & 250 & \textbf{RCT-only} & \textbf{0.916} & \textbf{0.000} & \textbf{0.276 $\pm$ 0.089} & \textbf{0.980} \\
shared-only & comparative & 0.2 & 500 & B-CALM & 0.930 & 0.060 & 0.143 $\pm$ 0.045 & 0.546 \\
shared-only & comparative & 0.2 & 500 & OSCAR/CALM & 0.921 & 0.020 & 0.150 $\pm$ 0.051 & 0.550 \\
shared-only & comparative & 0.2 & 500 & \textbf{RCT-only} & \textbf{0.921} & \textbf{0.000} & \textbf{0.196 $\pm$ 0.050} & \textbf{0.694} \\
shared-only & comparative & 0.2 & 500 & pooled & 0.277 & 0.920 & 0.262 $\pm$ 0.028 & 0.219 \\
shared-only & comparative & 0.2 & 500 & OS-only & 0.283 & 0.860 & 0.260 $\pm$ 0.028 & 0.232 \\
shared-only & comparative & 0.2 & 500 & \underline{CF} & \underline{0.791} & 0.900 & 0.281 $\pm$ 0.034 & 0.769 \\
shared-only & none &  & 250 & \underline{pooled} & \underline{0.895} & 0.000 & 0.068 $\pm$ 0.020 & 0.224 \\
shared-only & none &  & 250 & OS-only & 0.892 & 0.000 & 0.071 $\pm$ 0.021 & 0.231 \\
shared-only & none &  & 250 & \textbf{B-CALM} & \textbf{0.980} & \textbf{0.000} & \textbf{0.156 $\pm$ 0.056} & \textbf{0.683} \\
shared-only & none &  & 250 & CF & 0.962 & 0.080 & 0.157 $\pm$ 0.017 & 0.754 \\
shared-only & none &  & 250 & OSCAR/CALM & 0.918 & 0.040 & 0.226 $\pm$ 0.076 & 0.804 \\
shared-only & none &  & 250 & RCT-only & 0.936 & 0.000 & 0.265 $\pm$ 0.078 & 0.977 \\
shared-only & none &  & 500 & pooled & 0.927 & 0.000 & 0.061 $\pm$ 0.019 & 0.218 \\
shared-only & none &  & 500 & OS-only & 0.915 & 0.000 & 0.067 $\pm$ 0.021 & 0.231 \\
shared-only & none &  & 500 & \textbf{B-CALM} & \textbf{0.945} & \textbf{0.000} & \textbf{0.137 $\pm$ 0.054} & \textbf{0.516} \\
shared-only & none &  & 500 & OSCAR/CALM & 0.869 & 0.000 & 0.177 $\pm$ 0.059 & 0.548 \\
shared-only & none &  & 500 & CF & 0.965 & 0.180 & 0.151 $\pm$ 0.015 & 0.738 \\
shared-only & none &  & 500 & \underline{RCT-only} & \underline{0.887} & 0.000 & 0.212 $\pm$ 0.062 & 0.692 \\
\addlinespace
\multicolumn{9}{l}{Avg coverage: B-CALM 0.951; RCT-only 0.914; OSCAR/CALM 0.901; pooled 0.693; OS-only 0.694; CF 0.907} \\
\multicolumn{9}{l}{Avg neg-transfer: B-CALM 0.017; RCT-only 0.000; OSCAR/CALM 0.027; pooled 0.237; OS-only 0.227; CF 0.340} \\
\addlinespace
\bottomrule
\end{longtable}
}

\section{Additional experimental results}
\label{app:extra-experiments}

\subsection{Saturation effective sample size}
\label{app:saturation-ess}

\begin{table}[htbp]
\centering
\footnotesize
\setlength{\tabcolsep}{4pt}
\caption{Saturation diagnostic at the largest $n^o$. Empirical $\ess^o$ is computed from posterior precision after subtracting the RCT contribution.}
\label{tab:saturation}
\begin{tabular}{ccrrrr}
\toprule
$n^r$ & $\sigma_\Delta$ & Emp.\ $\ess^o$ & Th.\ $\ess^o$ & Width & Cov. \\
\midrule
250 & 0.1 & 402.7 & 396.8 & 0.521 & 0.989 \\
250 & 0.2 & 105.6 & 99.8  & 0.700 & 0.953 \\
250 & 0.4 & 30.1  & 25.0  & 0.796 & 0.929 \\
250 & inf & 4.0   & 0.0   & 0.932 & 0.923 \\
500 & 0.1 & 402.0 & 396.8 & 0.440 & 0.975 \\
500 & 0.2 & 105.1 & 99.8  & 0.535 & 0.941 \\
500 & 0.4 & 29.8  & 25.0  & 0.584 & 0.930 \\
500 & inf & 3.5   & 0.0   & 0.669 & 0.921 \\
\bottomrule
\end{tabular}
\end{table}

\subsection{Semi-synthetic benchmark}
\label{app:semi-synthetic}

Table~\ref{tab:semi} reports RMSE and 90\% coverage over $50$ replicates of the STAR-like benchmark; Appendix~\ref{app:semi-construction} specifies the construction routine, covariate marginals, imposed outcome surfaces, and bias functions. Figure~\ref{fig:semi-sensitivity} gives the accompanying sensitivity analysis: posterior CATEs at seven representative covariate profiles as the fixed comparative-bias scale $\sigma_\Delta$ varies from $0$ (pooled borrowing) to $\infty$ (RCT-only inference).

\begin{table}[t]
\centering
\small
\caption{STAR-like semi-synthetic benchmark with known potential outcomes (50 replicates; construction in Appendix~\ref{app:semi-construction}). MR-OSCAR represents the OSCAR/CALM family, whose members coincide in this harness; it is the family member run in this study.}
\label{tab:semi}
\begin{tabular}{lrr}
\toprule
Method & CATE RMSE & 90\% coverage \\
\midrule
B-CALM & 0.191 $\pm$ 0.057 & 0.936 \\
pooled & 0.194 $\pm$ 0.025 & 0.425 \\
OS-only & 0.216 $\pm$ 0.028 & 0.377 \\
CF & 0.218 $\pm$ 0.024 & 0.856 \\
MR-OSCAR & 0.228 $\pm$ 0.067 & 0.889 \\
RCT-only & 0.313 $\pm$ 0.078 & 0.888 \\
\bottomrule
\end{tabular}
\end{table}

\begin{figure}[H]
  \centering
  \includegraphics[width=0.55\linewidth]{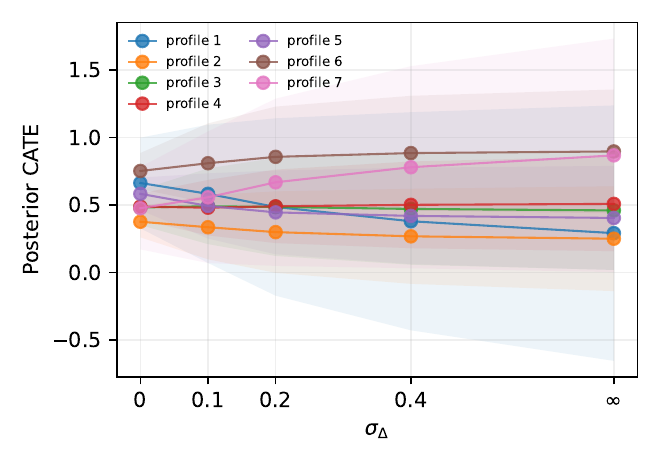}
  \caption{Sensitivity analysis on the STAR-like semi-synthetic benchmark (one replicate). Each line is the posterior CATE at one of seven representative covariate profiles (the first seven evaluation units), with shaded 90\% credible bands, as the fixed $\sigma_\Delta$ varies; estimates move between pooled borrowing ($\sigma_\Delta\to0$) and RCT-only inference ($\sigma_\Delta\to\infty$).}
  \label{fig:semi-sensitivity}
\end{figure}

\subsection{Ablations}
\label{app:ablation-table}

\begin{table}[t]
\centering
\small
\caption{B-CALM ablations on the balanced comparative-bias setting. Coverage and negative-transfer metrics are shown before RMSE to emphasize calibrated borrowing.}
\label{tab:ablation}
\begin{tabular}{lrrrr}
\toprule
Variant & Coverage & Neg. transfer & RMSE & Width \\
\midrule
Full B-CALM & 0.915 & 0.020 & 0.195 $\pm$ 0.047 & 0.670 \\
Deterministic encoders & 0.818 & 0.000 & 0.196 $\pm$ 0.047 & 0.506 \\
No b0 & 0.692 & 0.460 & 0.266 $\pm$ 0.057 & 0.551 \\
No b\_delta & 0.324 & 0.500 & 0.259 $\pm$ 0.026 & 0.279 \\
Point plus bootstrap & 0.259 & 0.660 & 0.292 $\pm$ 0.027 & 0.247 \\
\bottomrule
\end{tabular}
\end{table}

The ablation uses the balanced DGP at comparative-bias strength $0.4$ (an intermediate level outside both the sweep grid and the $\sigma_\Delta$ candidate set), $n^r=500$, $n^o=5{,}000$, and $50$ seeds per variant. \emph{Full B-CALM} is the empirical-Bayes--averaged fit of Algorithm~\ref{alg:bcalm}. \emph{No $b_0$} forces the baseline shift to zero ($\sigma_0=0$); \emph{No $b_\Delta$} pools the treatment contrast directly ($\sigma_\Delta=0$). The remaining two variants are covariance-scaled surrogates for reduced uncertainty propagation: \emph{Det.\ encoders} fits at $\sigma_\Delta=0.2$ and shrinks the posterior covariance by the factor $0.55$, representing the spread lost when representation uncertainty is ignored, and \emph{Point + bootstrap} fits the fully pooled model and shrinks its covariance by $0.80$, representing a point estimate with resampling-based uncertainty. The surrogates show how such pipelines under-cover without re-implementing each one.

\clearpage
\section{Posterior calibration and sensitivity plots}
\label{app:calibration}

\begin{figure}[H]
  \centering
  \includegraphics[width=0.48\linewidth]{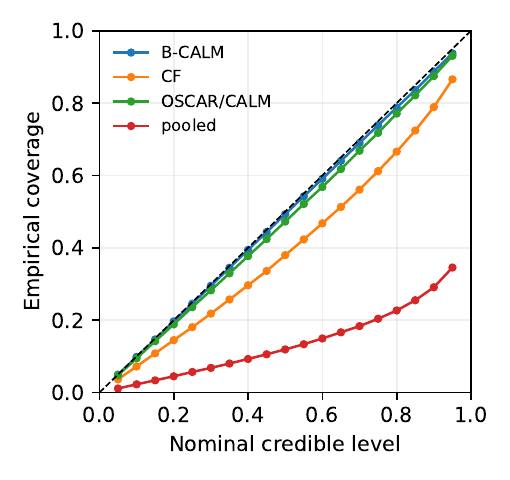}
  \caption{Posterior calibration on the balanced DGP under comparative bias: empirical coverage against nominal credible level. B-CALM tracks the diagonal while pooled and CF under-cover; the OSCAR/CALM family (R-OSCAR, MR-OSCAR, and CALM agree to within Monte Carlo noise and are plotted as one line) sits just below nominal at the 90\% mark.}
  \label{fig:calibration}
\end{figure}

\section{Real-data application diagnostics}
\label{app:gps-real}

This section reports the diagnostics supporting the pediatric-obesity external-control analysis in Section~\ref{sec:experiments}. Table~\ref{tab:gps-main} gives the full real-data baseline comparison, including Bayesian causal forest and OSCAR-family baselines omitted from the compact main-text table. Table~\ref{tab:gps-loo} checks sensitivity to masking one EHR control covariate at a time. Tables~\ref{tab:gps-bias} and~\ref{tab:gps-bias-susceptibility} report where the EHR controls disagree most with the trial control surface and how much hidden-baseline-bias risk each method's interval would absorb under selected and skeptical priors. Figure~\ref{fig:gps-real} visualizes the $\sigma_0$ borrowing path and the leave-one-covariate diagnostics, and Figure~\ref{fig:gps-bias-main} shows the corresponding baseline-bias recovery.

\begin{table}[H]
\centering
\footnotesize
\setlength{\tabcolsep}{3pt}
\caption{Real-data external-control augmentation. ATE on the trial covariate distribution; width reduction is relative to RCT-only. B-CALM uses $\sigma_0=0.05$, selected over a $\sigma_0$ grid (Appendix~\ref{app:reproducibility}). With one-arm EHR controls, $\sigma_\Delta$ is held diffuse and $\sigma_0$ is the active borrowing knob. The OSCAR/CALM row merges R-OSCAR, MR-OSCAR, and CALM, which coincide here (identical one-arm external-control calibration). The last column is the bias-susceptibility ratio under the skeptical $\sigma_0=0.6$ prior: 90\% range of the implied B-CALM-calibrated bias risk on each method's ATE, divided by that method's nominal width. A value above $1$ means a method's nominal interval is narrower than the bias it would absorb under a skeptical EHR-control prior.}
\label{tab:gps-main}
\begin{tabular}{@{}lccc@{}}
\toprule
Method & ATE [90\% CI] & Width (Red.) & Skep./w. \\
\midrule
B-CALM     & $-0.184\,[-0.336,-0.031]$ & $0.305\,(9.4\%)$  & $0.34$ \\
RCT-only   & $-0.174\,[-0.345,-0.008]$ & $0.337$           & $0.00$ \\
pooled     & $-0.191\,[-0.313,-0.074]$ & $0.239\,(28.9\%)$ & $1.05$ \\
CF         & $-0.162\,[-0.211,-0.112]$ & $0.099\,(70.6\%)$ & $2.91$ \\
BCF        & $-0.165\,[-0.242,-0.087]$ & $0.155\,(53.9\%)$ & --     \\
OSCAR/CALM & $-0.178\,[-0.331,-0.005]$ & $0.327\,(3.0\%)$  & $0.00$ \\
\bottomrule
\end{tabular}
\end{table}

\begin{table}[H]
\centering
\footnotesize
\setlength{\tabcolsep}{2pt}
\caption{Leave-one-covariate masking in the EHR controls. The final column checks whether the full-model 90\% CI contains the masked posterior mean.}
\label{tab:gps-loo}
\begin{tabular}{@{}lcccc@{}}
\toprule
Mask & ATE & 90\% CI & Width & Hit \\
\midrule
full & -0.184 & [-0.336, -0.031] & 0.305 & -- \\
sex & -0.184 & [-0.336, -0.029] & 0.306 & \checkmark \\
race\_ethnicity & -0.152 & [-0.303, -0.000] & 0.302 & \checkmark \\
language & -0.173 & [-0.325, -0.021] & 0.303 & \checkmark \\
insurance & -0.195 & [-0.361, -0.028] & 0.333 & \checkmark \\
wflz\_baseline & -0.186 & [-0.337, -0.031] & 0.306 & \checkmark \\
\bottomrule
\end{tabular}
\end{table}

\begin{table}[H]
\centering
\footnotesize
\setlength{\tabcolsep}{2pt}
\caption{Largest EHR-versus-RCT control-outcome discrepancies by posterior $|b_0(z)|$ under the skeptical commensurability prior $\sigma_0=0.6$. The skeptical prior gives $b_0$ enough room to absorb between-source baseline disagreement so that the diagnostic surfaces the strata where pooling would be most misleading; the trial-anchored ATE in Table~\ref{tab:gps-main} uses the BMA-selected $\sigma_0=0.05$ and yields much smaller posterior $b_0$ shrinkage.}
\label{tab:gps-bias}
\begin{tabular}{@{}p{0.46\textwidth}rcc@{}}
\toprule
Stratum & $n$ & Mean & 90\% CI \\
\midrule
Duke; No/other payer; base low & 3027 & -0.567 & [-1.178, 0.044] \\
Duke; Medicaid; base low & 542 & -0.501 & [-0.916, -0.086] \\
Duke; No/other payer; base mid & 3022 & -0.458 & [-1.002, 0.086] \\
\bottomrule
\end{tabular}
\end{table}

\begin{table}[H]
\centering
\footnotesize
\setlength{\tabcolsep}{2pt}
\caption{B-CALM-calibrated hidden-bias susceptibility in the GPS external-control analysis. The uniform column is the ATE change under a $+0.1$ WFLz shift of all EHR control outcomes. Bias-risk ranges are posterior 5th--95th percentiles of $\sum_s L_{m,s}b_{0,s}$ or row-level analogues.}
\label{tab:gps-bias-susceptibility}
\resizebox{\textwidth}{!}{%
\begin{tabular}{@{}lcccccccc@{}}
\toprule
Method & ATE & Nominal 90\% CI & Width & $+0.1$ shift & Selected risk & Sel./width & Skeptical risk & Skep./width \\
\midrule
B-CALM & -0.184 & [-0.336, -0.031] & 0.305 & -0.038 & [-0.038, 0.031] & 0.23 & [-0.064, 0.040] & 0.34 \\
RCT-only & -0.174 & [-0.345, -0.008] & 0.337 & 0.000 & [0.000, 0.000] & 0.00 & [0.000, 0.000] & 0.00 \\
pooled & -0.191 & [-0.313, -0.074] & 0.239 & -0.096 & [-0.094, 0.083] & 0.74 & [-0.142, 0.110] & 1.05 \\
CF & -0.162 & [-0.211, -0.112] & 0.099 & -0.088 & [-0.086, 0.121] & 2.08 & [-0.020, 0.268] & 2.91 \\
R-OSCAR & -0.178 & [-0.331, -0.005] & 0.327 & 0.000 & [-0.000, 0.000] & 0.00 & [-0.000, 0.000] & 0.00 \\
MR-OSCAR & -0.178 & [-0.331, -0.005] & 0.327 & 0.000 & [-0.000, 0.000] & 0.00 & [-0.000, 0.000] & 0.00 \\
CALM & -0.178 & [-0.331, -0.005] & 0.327 & 0.000 & [-0.000, 0.000] & 0.00 & [-0.000, 0.000] & 0.00 \\
B-CALM $\sigma_0=0.6$ & -0.170 & [-0.331, 0.004] & 0.336 & -0.000 & [-0.001, 0.000] & 0.00 & [-0.002, 0.001] & 0.01 \\
B-CALM $\sigma_0=\infty$ & -0.170 & [-0.340, 0.000] & 0.340 & -0.000 & [-0.000, 0.000] & 0.00 & [-0.000, 0.000] & 0.00 \\
\bottomrule
\end{tabular}%
}
\end{table}

\begin{figure}[t]
  \centering
  \begin{subfigure}[b]{0.48\textwidth}
    \centering
    \includegraphics[width=\linewidth]{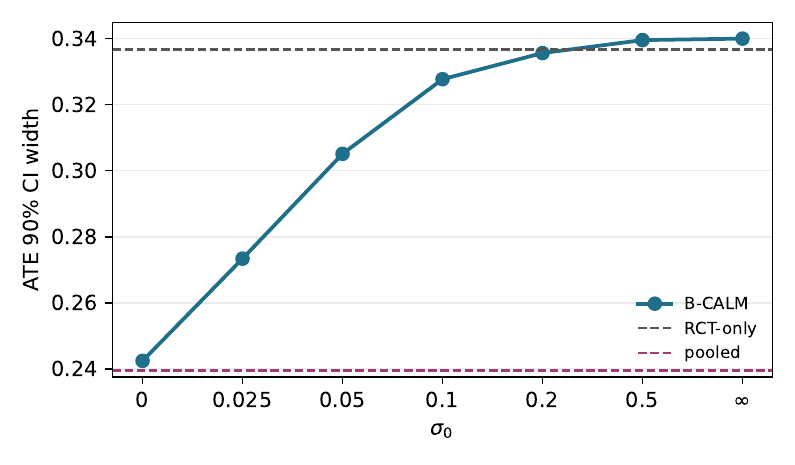}
    \caption{ATE interval width across $\sigma_0$.}
  \end{subfigure}\hfill
  \begin{subfigure}[b]{0.48\textwidth}
    \centering
    \includegraphics[width=\linewidth]{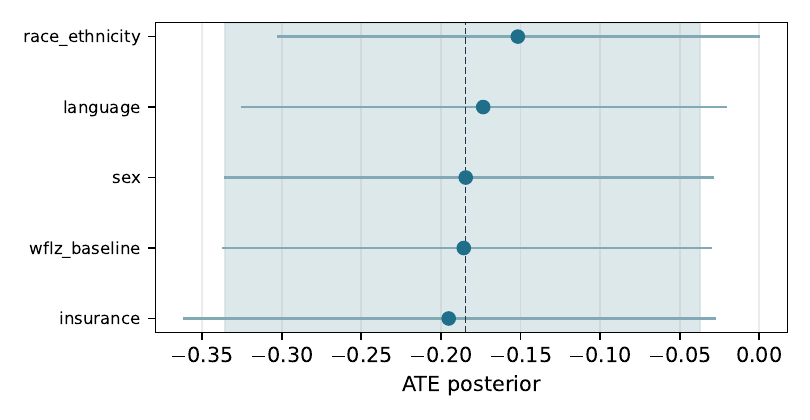}
    \caption{Leave-one-covariate masking.}
  \end{subfigure}
  \caption{Real-data external-control augmentation diagnostics. With one-arm EHR controls, $\sigma_0$ is the active borrowing knob and $\sigma_\Delta$ is held diffuse.}
  \label{fig:gps-real}
\end{figure}

\begin{figure}[t]
  \centering
  \includegraphics[width=0.90\textwidth]{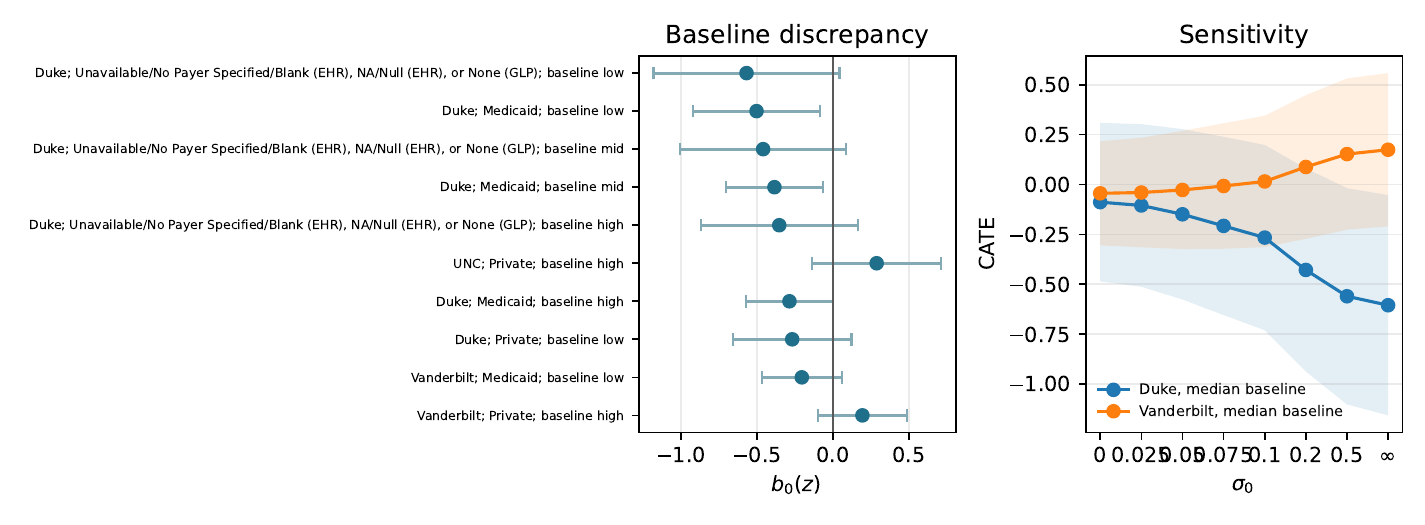}
  \caption{Baseline-bias recovery and $\sigma_0$ sensitivity in the pediatric-obesity application. With $A_j^o\equiv 0$ on every EHR row, $\sigma_0$ is the identified borrowing knob; $\sigma_\Delta$ is held diffuse.}
  \label{fig:gps-bias-main}
\end{figure}

\clearpage

\end{document}